%% file: main.tex
\documentclass{article}
\usepackage[utf8]{inputenc}
\usepackage[T1]{fontenc}
\usepackage{tgpagella}

\input{preamble.tex}

\title{Nearly Optimal Bounds for Sample-Based Testing and Learning of $k$-Monotone Functions\thanks{A preliminary version of this work appeared at RANDOM 2024.}}

\author{
Hadley Black\thanks{This work was partially completed while the author was a PhD student at UCLA supported by NSF award AF:Small 2007682, NSF Award: Collaborative Research Encore 2217033.} \\
CUNY, Baruch College\\
hadley.black@baruch.cuny.edu
}

\usepackage[ruled,vlined]{algorithm2e}
\usepackage{amsmath,amssymb,amsthm,verbatim,relsize,mathrsfs,hyperref}
\hypersetup{colorlinks = true,linkcolor = red, anchorcolor = blue, citecolor = red, filecolor = red,urlcolor = red}
\usepackage [english]{babel}
\usepackage [autostyle, english = american]{csquotes}
\usepackage{blindtext}
\usepackage{framed}
\usepackage{fullpage}
\usepackage[shortlabels]{enumitem}
\usepackage{graphicx}
\usepackage{subcaption}
\usepackage{subfiles}

\MakeOuterQuote{"}

\makeatletter
\newcommand\footnoteref[1]{\protected@xdef\@thefnmark{\ref{#1}}\@footnotemark}
\makeatother

\newcommand{\Expect}{\mathbb{E}}
\renewcommand{\Exp}{\Expect}

\renewcommand{\bt}{\pmb{t}}

\renewcommand{\norm}[1]{\left\lVert #1 \right\rVert}

\newcommand{\pr}{\mathbb{P}}
\renewcommand{\Pr}{\pr}

\newcommand{\Dyes}{\cD_{\texttt{yes}}}
\newcommand{\Dno}{\cD_{\texttt{no}}}
\newcommand{\Dterm}{\cD_{\texttt{term}}}
\newcommand{\mrk}{\mathcal{M}_{r,k}}

\renewcommand{\br}{\pmb{r}}

\newcommand{\block}{\mathsf{block}}
\newcommand{\blockpoint}{\mathsf{blockpoint}}

\newcommand{\cube}{\mathsf{cube}}

\begin{document}

\maketitle

%%% ABS 1

\begin{abstract} We study monotonicity testing of functions $f \colon \{0,1\}^d \to \{0,1\}$ using \emph{sample-based} algorithms, which are only allowed to observe the value of $f$ on points drawn independently from the uniform distribution. A classic result by Bshouty-Tamon (J. ACM 1996) proved that monotone functions can be learned with $\exp(\widetilde{O}(\min\{\frac{1}{\eps}\sqrt{d},d\}))$ samples and it is not hard to show that this bound extends to testing. Prior to our work the only lower bound for this problem was $\Omega(\sqrt{\exp(d)/\varepsilon})$ in the small $\eps$ parameter regime, when $\varepsilon = O(d^{-3/2})$, due to Goldreich-Goldwasser-Lehman-Ron-Samorodnitsky (Combinatorica 2000). Thus, the sample complexity of monotonicity testing was wide open for $\eps \gg d^{-3/2}$. We resolve this question, obtaining a nearly tight lower bound of $\exp(\Omega(\min\{\frac{1}{\eps}\sqrt{d},d\}))$ for all $\eps$ at most a sufficiently small constant. In fact, we prove a much more general result, showing that the sample complexity of $k$\emph{-monotonicity} testing and learning for functions $f \colon \{0,1\}^d \to [r]$ is $\exp(\Omega(\min\{\frac{rk}{\eps}\sqrt{d},d\}))$. For testing with one-sided error we show that the sample complexity is $\exp(\Omega(d))$.

Beyond the hypercube, we prove nearly tight bounds (up to polylog factors of $d,k,r,1/\eps$ in the exponent) of $\exp(\widetilde{\Theta}(\min\{\frac{rk}{\varepsilon}\sqrt{d},d\}))$ on the sample complexity of testing and learning measurable $k$-monotone functions $f \colon \mathbb{R}^d \to [r]$ under product distributions. Our upper bound improves upon the previous bound of $\exp(\widetilde{O}(\min\{\frac{k}{\varepsilon^2}\sqrt{d},d\}))$ by Harms-Yoshida (ICALP 2022) for Boolean functions ($r=2$). \end{abstract}

\thispagestyle{empty}
\newpage

\tableofcontents
\thispagestyle{empty}
\setcounter{page}{0}
\newpage

\newpage

\section{Introduction}

A function $f \colon \cX \to \RR$ over a partial order $\cP = (\cX,\preceq)$ is $k$-\emph{monotone} if there does not exist a chain of $k+1$ points $x_1 \prec x_2 \prec \cdots \prec x_{k+1}$ for which (a) $f(x_{i+1}) - f(x_i) < 0$ when $i$ is odd and (b) $f(x_{i+1}) - f(x_i) > 0$ when $i$ is even. When $k=1$, these are the \emph{monotone} functions, which are the non-decreasing functions with respect to $\preceq$. Monotone and $k$-monotone \emph{Boolean} functions over domains $\{0,1\}^d$, $[n]^d$, and $\RR^d$ have been the focus of a significant amount of research in property testing and computational learning theory. We give an overview of the literature in \Cref{sec:related}.

The field of property testing is concerned with the design and analysis of sub-linear time randomized algorithms for determining if a function has, or is far from having, some specific property. A key aspect in the definition of a property testing algorithm is the type of access it has to the function. Early works on property testing, e.g. \cite{RS96,GGR98}, focused on the notion of \emph{query-based} testers, which are allowed to observe the value of the function on any point of their choosing, and since then this has become the standard model. The weaker notion of \emph{sample-based} testers, which can only view the function on independent uniform samples, was also considered by \cite{GGR98} and has received some attention over the years, see e.g. \cite{KR00,BBBY12,FischerLV15,GR16,FH23,FH24}. Sample-based algorithms are considered more natural in many settings, for example in computational learning theory, where they are the standard model. In fact, sample-based testing and learning are closely related problems; given a learning algorithm, it is always possible to design a testing algorithm with the same sample complexity, up to an additive $\text{poly}(1/\eps)$ factor\footnote{See \Cref{lem:learning->testing} for a precise statement. Also, note that if the learning algorithm is \emph{proper}, then the time complexity is also preserved. If the learning algorithm is \emph{improper}, then there is a time complexity blow-up, but the sample complexity is still preserved.}. %This leads to the natural question of \emph{testing vs. learning}: when does testing require fewer samples than learning?

For many fundamental properties, there is still a large gap between how much we know in the query-based vs the sample-based models. Monotonicity (and $k$-monotonicity) is such a property; despite a vast body of research on query-based monotonicity testing over the hypercube $\{0,1\}^d$, the only work we know of which considers this problem in the sample-based model is \cite{GGLRS00}, who gave an upper bound of $O(\sqrt{2^d/\eps})$ and a matching lower bound for the case when $\eps = O(d^{-3/2})$ on the number of samples needed to test monotonicity of functions $f \colon \{0,1\}^d \to \{0,1\}$. The upper bound for learning monotone Boolean functions due to \cite{BshoutyT96,LRV22} also implies a testing upper bound of $\exp(\widetilde{O}(\frac{1}{\eps}\sqrt{d}))$. Thus, this question has been wide open for $\eps \gg d^{-3/2}$. %Our work aims to address this gap in the monotonicity testing literature.

Our work addresses this gap in the monotonicity testing literature, proving a lower bound which matches the learning upper bound for all $\eps$ at most some constant, up to a factor of $\log d$ in the exponent. More generally, we prove a nearly tight lower bound for $k$-monotonicity testing of functions, $f \colon \{0,1\}^d \to [r]$, i.e. functions with image size at most $r$. To round out our results, we also give an improved learning algorithm for $k$-monotone functions over $\RR^d$ under product distributions whose sample complexity matches our sample-based testing lower bound, up to poly-logarithmic factors in the exponent.

\subsection{Results} \label{sec:results}

Before explaining our results and the context for them, we first provide some terminology and basic notation. Given a domain $\cX$ and a distribution $\mu$ over $\cX$, we denote the Hamming distance between two functions $f,g \colon \cX \to \RR$ under $\mu$ by $d_{\mu}(f,g) = \pr_{x \sim \mu}[f(x) \neq g(x)]$. We say that $f$ is $\eps$-far from $k$-monotone if $d_{\mu}(f,g) \geq \eps$ for every $k$-monotone function $g$. The results in this paper pertain to sample-based testing and learning of $k$-monotone functions with respect to Hamming distance. We use the following terminology: %We will use the following terminology: %Before stating our results we lay out some terminology.

\begin{itemize}
    \item The \emph{example oracle} for $f$ under $\mu$, denoted by $EX(f,\mu)$, when queried, generates an example $(x,f(x))$ where $x$ is sampled according to $\mu$. 
    \item A \emph{sample-based $k$-monotonicity tester} under $\mu$ is a randomized algorithm which is given access to $EX(f,\mu)$ for an arbitrary input function $f$ and satisfies the following: (a) if $f$ is $k$-monotone, then the algorithm accepts with probability at least $2/3$, and (b) if $f$ is $\eps$-far from $k$-monotone, then the algorithm rejects with probability at least $2/3$. The tester has \emph{one-sided error} if in case (a) it accepts with probability $1$.
    \item A \emph{sample-based learning algorithm for $k$-monotone functions} under $\mu$ is a randomized algorithm which is given access to $EX(f,\mu)$ for an arbitrary $k$-monotone input function $f$ and outputs a hypothesis $h$ such that $d_{\mu}(f,h) \leq \eps$ with probability at least $1-\delta$. If left unspecified, $\delta = 1/3$.
\end{itemize}

In all of the above definitions if $\mu$ is unspecified, then it is the uniform distribution. Testing and learning are closely related problems; any sample-based learning algorithm can be used to construct a sample-based tester with the same sample complexity. We refer to this transformation as the testing-by-learning reduction and although this is not a new idea we provide a proof in \Cref{sec:testing-by-learning} for completeness.

Finally, we recall some important learning theory terminology. A learning algorithm for concept class $\cC$ is called \emph{proper} if it always outputs a hypothesis $h \in \cC$, and is called \emph{improper} if it is allowed to output arbitrary $h$. Given a function $f$, and a concept class $\cC$, let $d(f,\cC) = \min_{g \in \cC} d(f,g)$. An \emph{agnostic proper} learner is one which, given \emph{any} $f$ (not necessarily in $\cC$), outputs a hypothesis $h \in \cC$ for which $d(f,h) \leq d(f,\cC) + \eps$ with probability at least $1-\delta$.

\subsubsection{Sample-Based Testing and Learning on the Hypercube} 

The problem of \emph{learning} monotone Boolean functions over the hypercube $\{0,1\}^d$ was studied by \cite{BshoutyT96} who proved an upper bound\footnote{We remark that any function over $\{0,1\}^d$ can be learned exactly with $O(d 2^d) = \exp(O(d))$ samples by a coupon-collector argument. Combining this with the $\exp(O(\frac{1}{\eps}\sqrt{d} \log d))$ upper bound by \cite{BshoutyT96} yields $\exp(O(\min\{\frac{1}{\eps}\sqrt{d}\log d,d\}))$. We use this slightly clunkier notation involving the min to emphasize that our upper and lower bounds are nearly matching in all parameter regimes.} of $\exp(O(\min\{\frac{1}{\eps}\sqrt{d}\log d,d\}))$ for improper learning and very recently by \cite{LRV22,LV23} who obtained the same upper bound for agnostic proper learning. The improper learning upper bound was extended by \cite{BlaisCOST15} who showed an upper bound of $\exp(O(\min\{\frac{k}{\eps}\sqrt{d} \log d,d\}))$ and a nearly matching lower bound of $\exp(\Omega(\min\{\frac{k}{\eps}\sqrt{d},d\}))$ for learning $k$-monotone Boolean functions for any $k \geq 1$. The testing-by-learning reduction shows that their upper bound also holds for sample-based \emph{testing}. 
%On the lower bound side however little was known prior to our work, even for monotonicity testing ($k=1$). 
The only prior lower bound for sample-based testing that we're aware of is $\Omega(\sqrt{2^d/\eps})$ when $\eps = O(d^{-3/2})$ and $k=1$ \cite[Theorem 5]{GGLRS00}. Our main result is the following much more general lower bound for this problem, which we prove in \Cref{sec:LB}. 

%Our main result is the following lower bound for \emph{sample-based} $k$-monotonicity testing of functions over the hypercube, which we prove in \Cref{sec:LB}. 

\begin{theorem} [Testing Lower Bound] \label{thm:LB_2sided_samples_rk} There is an absolute constant $c > 0$ such that for all $\eps \leq c$, every sample-based $k$-monotonicity tester for functions $f\colon \{0,1\}^d \to [r]$ under the uniform distribution has sample complexity
\[
\exp\left(\Omega\left(\min\left\{\frac{rk}{\eps} \sqrt{d},d\right\}\right)\right) \text{.}
\]
\end{theorem}

Even for the special case of sample-based monotonicity testing of Boolean functions ($k=1$ and $r=2$), \Cref{thm:LB_2sided_samples_rk} is already a new result, which matches the upper bound for learning by \cite{BshoutyT96} and is the first lower bound to hold for $\eps \gg d^{-3/2}$. Moreover, our lower bound is much more general, holding for all $r,k$, and is optimal in all parameters, $d,r,k,\eps$, up to a $\log d$ factor in the exponent. We show a nearly matching upper bound in \Cref{thm:UB-hypercube}.

We also note that the testing-by-learning reduction implies that the same lower bound holds for \emph{learning} with samples. As we mentioned, this result was already known for Boolean functions (the $r=2$ case) \cite{BlaisCOST15}, but the general case of $r \geq 2$ was not known prior to our work\footnote{It is possible that the techniques from \cite{BlaisCOST15} could be extended to provide an alternative proof of \Cref{cor:LB_learning}, but we have not checked whether this is the case.}. 

\begin{corollary} [Learning Lower Bound] \label{cor:LB_learning} There is an absolute constant $c > 0$ such that for every $\eps \leq c$, every sample-based uniform-distribution learning algorithm for $k$-monotone functions $f\colon \{0,1\}^d \to [r]$ has sample complexity
\[
\exp\left(\Omega\left(\min\left\{\frac{rk}{\eps} \sqrt{d},d\right\}\right)\right) \text{.}
\]
\end{corollary}

On the upper bound side, a relatively straightforward argument extends the learning algorithm of \cite{BlaisCOST15} for Boolean $k$-monotone functions, to $k$-monotone functions with image size at most $r$. We give a short proof in \Cref{sec:UB-hypercube-proof}. This shows that our lower bounds in \Cref{thm:LB_2sided_samples_rk} and \Cref{cor:LB_learning} are tight up to a factor of $\log d$ in the exponent. 

\begin{theorem} [Learning Upper Bound for Hypercubes] \label{thm:UB-hypercube} There is a uniform-distribution learning algorithm for $k$-monotone functions $f \colon \{0,1\}^d \to [r]$ which achieves error at most $\eps$ with time and sample complexity
\[
\exp\left(O\left(\min\left\{\frac{rk}{\eps} \sqrt{d} \log d,d\right\}\right)\right) \text{.}
\]
\end{theorem}

%\begin{proof} \cite[Theorem 1.4]{BlaisCOST15} proved this result for the case of $r=2$. By standard arguments, their theorem can be strengthened slightly so that with probability at least $1-\delta$, the learner outputs a hypothesis with error at most $\eps$, and the time and sample complexity is only larger by a multiplicative factor of $\text{poly}(\eps^{-1},\delta^{-1})$. We will need this slightly stronger statement.

%For each $t \in [r]$, let $f_t \colon \{0,1\}^d \to \{0,1\}$ denote the thresholded Boolean function defined as $f_t(x) := \mathbf{1}(f(x) \geq t)$. Observe that for all $x \in \{0,1\}^d$ we have $f(x) = \text{argmax}_t \{ f_t(x) = 1 \}$. Thus, for each $t \in [r]$, run the learning algorithm of \cite{BlaisCOST15} with error parameters set to $\eps' := \eps/r$ and $\delta = 1/3r$ to obtain a hypothesis $h_t$. We have $\pr[d(h_t,f_t) > \eps/r] < 1/3r$. By a union bound, with probability at least $2/3$, every $t \in [r]$ satisfies $d(h_t,f_t) \leq \eps/r$. Moreover, if this holds then by another union bound we have $\pr_x[\exists t \in [r] \colon h_t(x) \neq f(x)] \leq \eps$. Thus, the hypothesis $h(x) := \text{argmax}_t \{ h_t(x) = 1 \}$ satisfies $d(h,f) \leq \eps$. The number of samples used is $\text{poly}(\eps^{-1},r) \cdot \exp(O(\min\{\frac{k}{\eps'} \sqrt{d},d\}))$ and this completes the proof. \end{proof}

The testing-by-learning reduction again gives us the following corollary.

\begin{corollary} [Testing Upper Bound for Hypercubes] \label{cor:UB-hypercube-testing} There is a sample-based $k$-monotonicity tester for functions $f \colon \{0,1\}^d \to [r]$ with sample complexity
\[
\exp\left(O\left(\min\left\{\frac{rk}{\eps} \sqrt{d} \log d,d\right\}\right)\right) \text{.}
\]
\end{corollary}

%\paragraph{Testing with One-Sided Error:} 

Lastly, we consider the problem of sample-based testing with \emph{one-sided error}. For monotonicity testing of functions $f \colon \{0,1\}^d \to \{0,1\}$ with \emph{non-adaptive queries}, we know that one-sided and two-sided error testers achieve the same query-complexity (up to $\text{polylog}(d,1/\eps)$ factors): there is a $\widetilde{O}(\sqrt{d}/\eps^2)$ one-sided error upper bound due to \cite{KhotMS18} and a $\widetilde{\Omega}(\sqrt{d})$ two-sided error lower bound due to \cite{Chen17}. We show that the situation is quite different for \emph{sample-based} monotonicity testing; while the sample complexity of two-sided error testers is $\exp(\widetilde{\Theta}(\min\{\frac{1}{\eps}\sqrt{d},d\}))$, one-sided error testers require $\exp(\Theta(d))$ samples for all $\eps$. 

%and give matching upper and lower bounds up to a constant factor in the exponent.

\begin{theorem} [Testing with One-Sided Error] \label{thm:one-sided} For every $d,r,k$, and $\eps > 0$, sample-based $k$-monotonicity testing of functions $f \colon \{0,1\}^d \to [r]$ with one-sided error requires $\exp(\Theta(d))$ samples. \end{theorem}

\subsubsection{Sample-Based Testing and Learning in Continuous Product Spaces}

Learning $k$-monotone Boolean-valued functions has also been studied over $\RR^d$ with respect to product measures by \cite{HarmsY22} who gave an upper bound of $\exp(\widetilde{O}(\min\{\frac{k}{\eps^2}\sqrt{d},d\}))$ where $\widetilde{O}(\cdot)$ hides polylog factors of $d,k$, and $1/\eps$. Our next result gives an upper bound which improves the dependence on $\eps$ from $1/\eps^2$ to $1/\eps$ in the exponent. By the same approach we used to generalize the upper bound in \Cref{thm:UB-hypercube} to arbitrary $r \geq 2$, we get the same generalization for product spaces. We obtain the following upper bound which matches our lower bound for $\{0,1\}^d$ in \Cref{thm:LB_2sided_samples_rk} up to polylog factors of $d,k,r$, and $1/\eps$. We say that a function $f \colon \RR^d \to [r]$ is \emph{measurable} if the set $f^{-1}(i)$ is measurable for every $i \in [r]$.

\begin{theorem} [Learning Upper Bound for Product Spaces] \label{thm:UB-learner} Given an arbitrary product measure $\mu$, there is a learning algorithm under $\mu$ for measurable $k$-monotone functions $f \colon \RR^d \to [r]$ with time and sample complexity
%\[
%\exp\left(\widetilde{O}_{d,r,k,1/\eps}\left(\min\left\{\frac{rk}{\eps} \sqrt{d},d\right\}\right)\right) \text{.}
%\]
\[
\exp\left(\widetilde{O}\left(\min\left\{\frac{rk}{\eps} \sqrt{d},d\right\}\right)\right) \text{.}
\]
The $\widetilde{O}(\cdot)$ hides polylogarithmic dependencies on $d,r,k$, and $1/\eps$.
\end{theorem}

We prove \Cref{thm:UB-learner} in \Cref{sec:main-monotone-k}. Once again the testing-by-learning reduction gives us the following corollary for sample-based testing.

\begin{corollary} [Testing Upper Bound for Product Spaces] \label{cor:UB-rk-testing-informal} Given an arbitrary product measure $\mu$, there is a $k$-monotonicity tester for measurable functions $f \colon \RR^d \to [r]$ under $\mu$ with sample complexity
%\[
%\exp\left(\widetilde{O}_{d,r,k,1/\eps}\left(\min\left\{\frac{rk}{\eps} \sqrt{d},d\right\}\right)\right) \text{.}
%\]
\[
\exp\left(\widetilde{O}\left(\min\left\{\frac{rk}{\eps} \sqrt{d},d\right\}\right)\right) \text{.}
\]
The $\widetilde{O}(\cdot)$ hides polylogarithmic dependencies on $d,r,k$, and $1/\eps$.
\end{corollary}

\subsection{Proof Overviews} 

In this section we give an overview of our proofs for \Cref{thm:LB_2sided_samples_rk} and \Cref{thm:UB-learner}. 

\subsubsection{The Testing Lower Bound for Hypercubes} \label{sec:lower-bound-sketch}

Our proof of \Cref{thm:LB_2sided_samples_rk} uses a family functions known as \emph{Talagrand's random DNFs} introduced by \cite{Tal96} which have been used by \cite{BeBl16} and \cite{Chen17} to prove lower bounds for monotonicity testing of Boolean functions $f \colon \{0,1\}^d \to \{0,1\}$ against adaptive and non-adaptive query-based testers. Very recently, they have also been used to prove lower bounds for tolerant monotonicity testing \cite{CDLNS23} and for testing convexity of sets in $\{-1,0,1\}^d$ \cite{BBH24}. %Specifically, \cite{Chen17} established a $\widetilde{\Omega}(\sqrt{d})$ non-adaptive lower bound, which is known to be tight up to a $\text{poly}(\log d)$ factor, and a $\widetilde{\Omega}(d^{1/3})$ adaptive lower bound. We show that a similar type of construction can be used to prove much larger lower bounds against the weaker class of \emph{sample-based} testers. Moreover, our lower bound holds for the more general problem of testing $k$-monotonicity of functions $f \colon \{0,1\}^d \to [r]$. 

To understand our construction, let us first consider the special case of monotonicity of Boolean functions, i.e. $k=1$ and $r=2$. We think of a DNF term as a point $t \in \{0,1\}^d$ which is said to be satisfied by $x \in \{0,1\}^d$ if $t \preceq x$, where $\preceq$ denotes the standard bit-wise partial order over $\{0,1\}^d$. The \emph{width} of a term $t$ is its Hamming weight, $|t|$, and the width of a DNF is the max width among its terms. Consider $N$ randomly chosen terms $t^{1},\ldots,t^{N}$ each of width $|t^j| = w$. We will see later how to choose $N$ and $w$. Let $B := \{x \colon \frac{d}{2} \leq |x| \leq \frac{d}{2} + \eps\sqrt{d}\}$ and for each $j \in [N]$, let 
\[
U_j := \{x \in B \colon t^j \preceq x \text{ and } t^{j'} \not\preceq x \text{ for all } j' \neq j\}
\]
be the set of points in $B$ which satisfy $t^j$ and no other terms. Let $U := \bigcup_{j \in [N]} U_j$. Now observe that any two points lying in different $U_j$'s are \emph{incomparable} and therefore independently embedding an arbitrary monotone function into each $U_j$ will result in a function which globally is monotone if one defines the function outside of $U$ appropriately. Using this fact we can define two distributions $\Dyes$ and $\Dno$ as follows. Let $A$ denote the set of points in $x \in \{0,1\}^d$ for which either $|x| > \frac{d}{2} + \eps \sqrt{d}$ or $x \in B$ and $t^j,t^{j'} \preceq x$ for two different terms $j \neq j'$.

\begin{itemize}
    \item $f \sim \Dyes$ is drawn by setting $f(x) = 1$ if and only if $x \in A \cup \left(\bigcup_{j \in T} U_j\right)$ where $T \subseteq [N]$ contains each $j \in [N]$ with probability $1/2$, independently. Such a function is always monotone.
    \item $f \sim \Dno$ is drawn by setting $f(x) = 1$ if and only if $x \in A \cup R$ where $R$ contains each $x \in U$ with probability $1/2$, independently. Such a function will be $\Omega(|U| \cdot 2^{-d})$-far from monotone with probability $\Omega(1)$ since its restriction with $U$ is uniformly random.
\end{itemize}
Now, each $x \in U$ satisfies $\pr_{f \sim \Dyes}[f(x) = 1] = \pr_{f \sim \Dno}[f(x) = 1] = 1/2$ and for both distributions the events $f(x) = 1$ and $f(y) = 1$ are independent when $x,y$ lie in different $U_j$'s. Therefore, any tester will need to see at least two points from the same $U_j$ to distinguish $\Dyes$ and $\Dno$. Roughly speaking, by birthday paradox this gives a $\Omega(\sqrt{N})$ lower bound on the number of samples. The lower bound is thus determined by the maximum number of terms $N$ that can be used in the construction for which $|U| = \Omega(\eps 2^d)$.

So how are $N$ and $w$ chosen? By standard concentration bounds, we have $|B| = \Omega(\eps 2^d)$ and observe that a point $x \in B$ satisfies a random term with probability exactly $(|x|/d)^w$. We need $U$ to contain a \emph{constant fraction} of $B$, i.e. we need $x$ to satisfy exactly $1$ term with constant probability. The expected number of satisfied terms is $N \cdot (|x|/d)^w$ and, roughly speaking, we need this value to be $\Theta(1)$ for all $x \in B$. Applying this constraint to the case when $|x| = d/2$ forces us to pick $N \approx 2^w$. Now when $|x| = d/2 + \eps\sqrt{d}$, the expected number of satisfied terms is $N \cdot 2^{-w} \cdot (1+2\eps/\sqrt{d})^w \approx (1+2\eps/\sqrt{d})^w$ and we are forced to choose $w \approx \sqrt{d}/\eps$. The lower bound for sample-based monotonicity testing of $f \colon \{0,1\}^d \to \{0,1\}$ is then $\Omega(\sqrt{N}) \approx \exp(\Omega(\sqrt{d}/\eps))$.

Let us now think about generalizing this construction to testing $k$-monotonicity of functions $f \colon \{0,1\}^d \to [r]$. The moral of the above argument is that the permitted number of terms is controlled by the number of distinct Hamming weights in the set $B$. We observe that for larger values of $k$ and $r$ we can partition $B$ into $k(r-1)$ blocks as $B := B_1 \cup B_2 \cup \cdots \cup B_{k(r-1)}$ each with a window of Hamming weights of size only $\frac{\eps\sqrt{d}}{k(r-1)}$. We are able to essentially repeat the above construction independently within each block wherein we can set $w \approx \frac{k(r-1)\sqrt{d}}{\eps}$ and consequently $N \approx 2^{\frac{k(r-1)\sqrt{d}}{\eps}}$.

For each block $i \in [k(r-1)]$, the random Talagrand DNF within block $B_i$ is defined analogously to the above construction, except that it assigns function values from $\{i \bmod{(r-1)}, i \bmod{(r-1)} + 1\}$, instead of $\{0,1\}$. %See \Fig{layered-DNF} for an illustration. 
Since there are $k(r-1)$ blocks in total, the distribution $\Dyes$ only produces $k$-monotone functions. At the same time, a function $f \sim \Dno$ assigns uniform random $\{a,a+1\}$ values within each block $B_{m(r-1) + a}$. This results in a large number of long chains through $B_{a} \cup B_{(r-1) + a} \cup \cdots \cup B_{(k-1)(r-1)+a}$ which alternate between function value $a$ and $a+1$. Considering the union of all such chains for $a = 0,1,\ldots,r-2$ shows that $f$ is $\Omega(\eps)$-far from $k$-monotone with probability $\Omega(1)$.

%\begin{figure}
%    \hspace*{-2.3cm}
%    \includegraphics[scale=0.58]{figurev2.pdf}
%    \caption{An illustration of the construction used in our proof of \Cref{thm:LB_2sided_samples_rk}. The image represents the set of points in the hypercube $\{0,1\}^d$ with Hamming weight in the interval $[\frac{d}{2},\frac{d}{2}+\eps\sqrt{d})$, increasing from bottom to top. The numbers on the left denote the Hamming weight of the points lying in the adjacent horizontal line. The $B_i$ blocks are the sets of points contained between two adjacent horizontal lines. Each orange shaded region within $B_i$ represents the set of points satisfied by a term $t^{i,j}$. The blue numbers represent the value that functions in the support of $\Dyes$ and $\Dno$ can take. We have used the notation "$r-1,2$" as shorthand for $r-2,r-1$.}
%    \label{fig:layered-DNF}
%\end{figure}

\subsubsection{The Learning Upper Bound for Product Spaces} \label{sec:sketch}

As we discussed in \Cref{sec:results}, it suffices to prove \Cref{thm:UB-learner} for the case of $r=2$, i.e. learning functions $f \colon \RR^d \to \{\pm 1\}$ under a product measure $\mu$. We use a downsampling technique to reduce this problem to learning a discretized proxy of $f$ over a hypergrid $[N]^d$ where $N = \Theta(kd/\eps)$ under agnostic label noise at rate $O(\eps)$. This technique has been used in previous works \cite{GKW19,BlackCS20,HarmsY22} and our proof borrows many technical details from \cite{HarmsY22}.

Next, for $N$ which is a power of $2$, we observe that a $k$-monotone function $f \colon [N]^d \to \{\pm 1\}$ can be viewed as a $k$-monotone function over the hypercube $\{\pm 1\}^{d \log N}$ by mapping each point $x \in [N]^d$ to its bit-representation. We can then leverage a result of \cite{BlaisCOST15} which shows that all but a $\eps$-fraction of the mass of the Fourier coefficients of $k$-monotone Boolean functions $f \colon \{0,1\}^d \to \{0,1\}$ is concentrated on the terms with degree at most $\frac{k\sqrt{d}}{\eps}$. We can then use the low-degree algorithm introduced by \cite{LinialMN93} which was shown to work even in the agnostic setting by \cite{agnostichalfspaces08}.

\subsection{Discussion and Open Questions}

Our results for sample-based testing and learning over the hypercube are tight up to a $\log d$ factor in the exponent. Our upper bound for product spaces matches the lower bound for hypercubes only up to \emph{polylog} factors of $d,k,r,1/\eps$ in the exponent. In particular, the upper bound for product spaces goes to $\infty$ as any one of the parameters $r$, $k$, or $1/\eps$ grow to $\infty$, whereas the lower bound for the hypercube can be at most $\exp(\Theta(d))$ simply because $|\{0,1\}^d| = 2^d$ and so any function $f \colon \{0,1\}^d \to \RR$ can be learned \emph{exactly} with $\exp(O(d))$ samples. It seems intuitive that sample-based testing and learning of $k$-monotone functions over $[n]^d$ should require $n^{\Omega(d)}$ samples as either of the parameters $k$ or $r$ approaches $\infty$. A corollary of such a result would be that the sample-complexity of these problems for $f \colon \RR^d \to [r]$ grow to $\infty$ as $k$ or $r$ approach $\infty$. Moreover, if this is true, then $k$-monotonicity of functions $f \colon \RR^d \to \RR$ is not testable with a finite number of samples. Our results do not address this and it would be interesting to investigate this further.

\begin{question} Is there a lower bound for sample-based $k$-monotonicity testing of functions $f \colon [n]^d \to [r]$ which approaches $n^{\Omega(d)}$ as $r$ or $k$ go to $\infty$? \end{question}

\subsection{Related Work} \label{sec:related}

%\subsubsection{Monotonicity and $k$-Monotonicity}

Monotone functions and their generalization to $k$-monotone functions have been extensively studied within property testing and learning theory over the last 25 years. We highlight some of the results which are most relevant to our work. Afterwards, we discuss some selected works on sample-based property testing.

\paragraph{Sample-based monotonicity testing:} Sample-based monotonicity testing of Boolean functions over the hypercube, $\{0,1\}^d$, was considered by \cite{GGLRS00} (see \cite[Theorems 5 and 6]{GGLRS00}) who gave an upper bound of $O(\sqrt{2^d/\eps})$ and a lower bound of $\Omega(\sqrt{2^d/\eps})$ for $\eps = O(d^{-3/2})$. Sample-based monotonicity testing over general partial orders was studied by \cite{FLNRRS02} who gave a $O(\sqrt{N/\eps})$ one-sided error tester for functions $f \colon D \to \RR$ where $D$ is any partial order on $N$ elements. Sample-based monotonicity testing of functions on the line $f \colon [n] \to [r]$ was studied by \cite{PallavoorRV18} who gave a one-sided error upper bound of $O(\sqrt{r/\eps})$ and a matching lower bound of $\Omega(\sqrt{r})$ for all sample-based testers.  

\paragraph{Query-based monotonicity testing:} Monotonicity testing has been extensively studied in the standard query model \cite{chen2025boolean,Ras99,EKK+00,GGLRS00,DGLRRS99,LR01,FLNRRS02,HK03,AC04,HK04,ACCL04,E04,SS08,Bha08,BCG+10,FR,BBM11,RRSW11,BGJ+12,ChSe13,ChSe13-j,ChenST14,BeRaYa14,BlRY14,ChenDST15,ChDi+15,KMS15,BeBl16,Chen17,BlackCS18,PallavoorRV18,BlackCS20,HarmsY22,BKR23,BKKM23,BCS23stoc,BCS23focs,CDLNS23}. %The last 10 years have seen exciting developments for testing monotonicity of Boolean functions over hypercubes $\{0,1\}^d$, hypergrids $[n]^d$, and $\RR^d$ under product measures. 
When discussing these works we treat $\eps$ as a small constant for brevity. For $f \colon \{0,1\}^d \to \{0,1\}$, the non-adaptive query complexity has been established at $
\widetilde{\Theta}(\sqrt{d})$ \cite{KhotMS18,Chen17} with a (very recent) nearly-matching adaptive lower bound of $\Omega(d^{1/2-c})$, for any constant $c > 0$ \cite{chen2025boolean}. %This gap for adaptive monotonicity testing of Boolean functions is still an outstanding open question. 
For $f \colon [n]^d \to \{0,1\}$ and $f \colon \RR^d \to \{0,1\}$ under product measures, a recent result of \cite{BCS23focs} established a non-adaptive upper bound of $d^{1/2+o(1)}$. For functions $f \colon \{0,1\}^d \to [r]$, \cite{BKR23} showed upper and lower bounds of $\widetilde{\Theta}(\min(r\sqrt{d},d))$ for non-adaptive, one-sided error testers and there is a general (adaptive) lower bound of $\Omega(\min(d,r^2))$ due to \cite{BBM11}. For real-valued functions $f \colon [n]^d \to \RR$, the query complexity is known to be $\Theta(d \log n)$. The upper bound is non-adaptive \cite{ChSe13} and the lower bound holds even for adaptive testers \cite{ChSe14}.

\paragraph{$k$-Monotonicity testing:} The generalization to $k$-monotonicity testing has also been studied in the standard query model by \cite{GKW19,CGG0W19}. These works show that the query-complexity of non-adaptive one-sided error $k$-monotonicity testing is $\exp(\widetilde{\Theta}(\sqrt{d}))$ for all $k \geq 2$, demonstrating an interesting separation between (1-)monotonicity and 2-monotonicity.

\paragraph{Learning monotone functions:} Monotone Boolean functions $f \colon \{0,1\}^d \to \{0,1\}$ were studied in the context of learning theory by \cite{BshoutyT96} who showed that they can be (improperly) learned to error $\eps$ under the uniform distribution with $\exp(\widetilde{O}(\frac{1}{\eps}\sqrt{d}))$ time and samples. Very recent works \cite{LRV22,LV23} have given \emph{agnostic proper} learning algorithms with the same complexity.

\paragraph{Learning $k$-monotone functions:} The result of \cite{BshoutyT96} was generalized by \cite{BlaisCOST15} who gave upper and lower bounds of $\exp(\widetilde{\Theta}(\frac{k}{\eps}\sqrt{d}))$ for learning $k$-monotone Boolean functions $f \colon \{0,1\}^d \to \{0,1\}$. For Boolean functions over hypergrids $f \colon [n]^d \to \{0,1\}$, \cite{CGG0W19} gave an upper bound of $\exp(\widetilde{O}(\min(\frac{k}{\eps^2}\sqrt{d},d)))$ where $\widetilde{O}(\cdot)$ hides polylog factors of $d,k,1/\eps$. This result was generalized to functions $f \colon \RR^d \to \{0,1\}$ under product measures by \cite{HarmsY22}.

\paragraph{Sample-based property testing:} The notion of sample-based property testing was first presented and briefly studied by \cite{GGR98}. Broader studies of sample-based testing and its relationship with query-based testing have since been given by \cite{FischerGL14, FischerLV15,GR16}. A characterization of properties which are testable with a constant number of samples was given by \cite{BlaisY19}.

As we mentioned, sample-based algorithms are the standard model in learning theory, and learning requires at least as many samples as testing for every class of functions. Thus, it is natural to ask, when is testing \emph{easier} than learning in terms of sample complexity? This question is referred to as \emph{testing vs learning} and has been studied by \cite{KR00} and more recently by \cite{BFH21,FH23,FH24}. %Our results show that for monotonicity and $k$-monotonicity, testing is not easier than learning, up to constant factors in the exponent. 

There has also been work studying models that interpolate between query-based and sample-based testers. For instance, \cite{BBBY12} introduced the notion of \emph{active testing}, where the tester may make queries, but only on points from a polynomial-sized batch of unlabeled samples drawn from the underlying distribution. This was inspired by the notion of \emph{active learning} which considers learning problems under this access model.

Sample-based convexity testing of sets over various domains has also seen some recent attention \cite{ChenFSS17,BermanMR19b,BermanMR19a,BBH24}.

\subsection{Learning Functions with Bounded Image Size: Proof of \Cref{thm:UB-hypercube}} \label{sec:UB-hypercube-proof}

In this section we give a short proof showing that the learning algorithm of \cite{BlaisCOST15} can be extended in a relatively straightforward manner to functions $f \colon \{0,1\}^d \to [r]$ by increasing the sample-complexity by a factor of $r$ in the exponent. 

\begin{proof}[Proof of \Cref{thm:UB-hypercube}] \cite[Theorem 1.4]{BlaisCOST15} proved this result for the case of $r=2$. In particular, they show that there is a sample-based learning algorithm which given an arbitrary $k$-monotone Boolean function $f$, outputs $h$ such that $\Pr_h[d(f,h) > \eps] < \delta$ using $\ln(1/\delta) \cdot \exp\big(O\big(\min\big\{\frac{rk}{\eps} \sqrt{d} \log d,d\big\}\big)\big)$ queries\footnote{Their result (Thm 1.4 of \cite{BlaisCOST15}) is stated for constant $\delta$, but can be easily extended to arbitrary $\delta$ with the stated query complexity by replacing Thm 3.1 in their proof with the low-degree algorithm stated for general $\delta$.} to the example oracle, $EX(f)$. We will make use of this result.

For each $t \in [r]$, let $f_t \colon \{0,1\}^d \to \{0,1\}$ denote the thresholded Boolean function defined as $f_t(x) := \mathbf{1}(f(x) \geq t)$. Observe that for all $x \in \{0,1\}^d$ we have $f(x) = \text{argmax}_t \{ f_t(x) = 1 \}$, and if $f$ is $k$-monotone, then so is $f_t$ for every $t \in [r]$. Thus, for each $t \in [r]$, run the learning algorithm of \cite{BlaisCOST15} with error parameters set to $\eps' := \eps/r$ and $\delta = 1/3r$ to obtain a hypothesis $h_t$. We have $\pr[d(h_t,f_t) > \eps/r] < 1/3r$. By a union bound, with probability at least $2/3$, every $t \in [r]$ satisfies $d(h_t,f_t) \leq \eps/r$. Moreover, if this holds then by another union bound we have $\pr_x[\exists t \in [r] \colon h_t(x) \neq f_t(x)] \leq \eps$. Thus, the hypothesis $h(x) := \text{argmax}_t \{ h_t(x) = 1 \}$ satisfies $d(h,f) \leq \eps$. The number of samples used is $\ln(1/\delta) \cdot \exp(O(\min\{\frac{k}{\eps'} \sqrt{d} \log d,d\})) = \exp(O(\min\{\frac{rk}{\eps} \sqrt{d} \log d,d\}))$ and this completes the proof. \end{proof}

\section{Preliminaries on $k$-Monotonicity}

We use the notation $[n] := \{0,1,\ldots,n-1\}$. Throughout the paper, we use the convention that random variables are denoted by bold symbols.

\begin{definition} \label{def:alt-chains} Given a poset $\cP = (\cX,\preceq)$ and a function $f \colon \cX \to \mathbb{R}$, an $m$-alternating chain is a sequence of points $x_1 \prec x_2 \prec \cdots \prec x_{m}$ such that for all $i \in \{1,\ldots,m-1\}$,
\begin{enumerate}
    \item $f(x_{i+1}) - f(x_i) < 0$ when $i$ is odd, and
    \item $f(x_{i+1}) - f(x_i) > 0$ when $i$ is even.
\end{enumerate}
\end{definition}

\begin{definition} [$k$-monotonicity] \label{def:k-mono} For a poset $\cP = (\cX,\preceq)$, a function $f \colon \cX \to \mathbb{R}$ is called $k$-monotone if it does not have any $(k+1)$-alternating chains.
\end{definition}

Let $\cM_{\cP,k}$ denote the set of all $k$-monotone functions $f \colon \cX \to \mathbb{R}$ over the poset $\cP = (\cX,\preceq)$. The Hamming distance between two functions $f,g \colon \cX \to \mathbb{R}$ is $d(f,g) = |\cX|^{-1} \cdot |\{x \in \cX \colon f(x) \neq g(x)\}|$. The distance to $k$-monotonicity of $f$ is denoted by $\eps(f,\cM_{\cP,k}) := \min_{g \in \cM_{\cP,k}} d(f,g)$. The following claim is our main tool for lower bounding the distance to $k$-monotonicity.

\begin{claim} \label{clm:far} Let $f \colon \cX \to \mathbb{R}$ and $k' \geq 3k$ be an integer. Let $\cC \subset \cX^{k'}$ be a collection of disjoint $k'$-alternating chains for $f$. Then 
\[
    \eps(f,\cM_{\cP,k}) \geq \frac{1}{3|\cX|} \cdot \left| \bigcup_{C \in \cC} C \right| \text{.}
\]
%\[
%\eps(f,\cM_{\cP,k}) \geq \frac{(\alpha-2) \cdot k \cdot |\cC|}{|\cX|} \geq \frac{1}{3|\cX|} \cdot \left| \bigcup_{C \in \cC} C \right| \text{.}
%\]
\end{claim}

\begin{proof} Observe that every $k$-monotone function $g \in \cM_{\cP,k}$ has the following property: for every $C = (x_1, x_2, \ldots, x_{k'}) \in \cC$, the sequence 
\[
\big(1, g(x_2) - g(x_1), g(x_3) - g(x_2), \ldots, g(x_{k'}) - g(x_{k'-1}) \big)
\]
changes sign at most $k-1$ times, whereas the sequence 
\[
\big(1, f(x_2) - f(x_1), f(x_3) - f(x_2), \ldots, f(x_{k'}) - f(x_{k'-1}) \big)
\]
changes sign exactly $k' - 1$ times. We have prepended a $1$ so that the first sign change occurs as soon as the function value decreases. Now, changing $f(x_i)$ can only reduce the number of times the sequence changes sign by at most $2$ and so $|\{i \colon f(x_i) \neq g(x_i)\}| \geq \frac{k'-k}{2}$. Summing over all chains in $\cC$ and normalizing yields 
\[
d(f,g) \geq \frac{k'-k}{2} \cdot \frac{|\cC|}{|\cX|} \geq \frac{k'}{3} \cdot \frac{|\cC|}{|\cX|} \geq \frac{1}{3|\cX|} \cdot \left| \bigcup_{C \in \cC} C \right|
\]
where the second inequality follows from $k \leq k'/3$ and the third inequality is due to the fact that the chains in $\cC$ are all disjoint and each of size $k'$. This completes the proof since this inequality holds for all $g \in \cM_{\cP,k}$. \end{proof}

We use the notation $\mrk$ to denote the set of all $k$-monotone functions $f \colon \{0,1\}^d \to [r]$ over the hypercube whose image has at most $r$ distinct values.

\section{Lower Bound for Sample-Based Testers} \label{sec:LB}

In this section we prove \Cref{thm:LB_2sided_samples_rk}, our lower bound on the sample-complexity of testing $k$-monotonicity of functions $f \colon \{0,1\}^d \to [r]$. We refer the reader to \Cref{sec:lower-bound-sketch} for a discussion of our main ideas and a proof sketch for the special case of $k=1$ and $r=2$, i.e. \emph{monotone Boolean} functions. Our proof follows the standard approach of defining a pair of distributions $\Dyes,\Dno$ over functions $f \colon \{0,1\}^d \to [r]$ which satisfy the following:

\begin{itemize}
    \item $\Dyes$ is supported over $k$-monotone functions.
    \item Functions drawn from $\Dno$ are typically $\Omega(\eps)$-far from $k$-monotone: $\pr_{f \sim \Dno}[\eps(f, \cM_{r,k}) = \Omega(\eps)] = \Omega(1)$. 
    \item The distributions over labeled examples from $\Dyes$ and $\Dno$ are close in TV-distance.
\end{itemize}

Our construction uses a generalized version of a family of functions known as random Talagrand DNFs, which were used by \cite{BeBl16} and \cite{Chen17} to prove lower bounds for testing monotonicity of Boolean functions with adaptive and non-adaptive queries. %The original construction corresponds to the case of $r = 2$ and $k = 1$, i.e. Boolean-valued monotone functions and we show that a . 

Let $r,k$ satisfy $rk \leq \frac{\eps \sqrt{d}}{24300}$. (Recall our goal is to prove the lower bound $\exp(\Omega(\min\{\frac{rk}{\eps}\sqrt{d},d\}))$ and so it suffices to restrict our attention to such values of $r,k$.) For convenience, we will assume that $\frac{k(r-1)}{2\eps}$ and $\sqrt{d}$ are integers and that $\frac{k(r-1)}{\eps}$ divides $\sqrt{d}$. Let $L_{\ell} := \left\{x \in \{0,1\}^d \colon |x| = \ell\right\}$ denote the $\ell$'th Hamming level of the hypercube. We partition $\bigcup_{\ell \in [0,\eps\sqrt{d})} L_{d/2 + \ell}$ into $k(r-1)$ blocks as follows. For each $i \in [k(r-1)]$, define
\[
B_i = \bigcup_{\ell = i \cdot \frac{\eps\sqrt{d}}{k(r-1)}}^{(i+1) \cdot \frac{\eps\sqrt{d}}{k(r-1)} - 1} L_{\frac{d}{2} + \ell}\text{.}
\]
The idea of our proof is to define a random DNF within each $B_i$. The \emph{width} of each DNF will be set\footnote{Since $\frac{k(r-1)}{2\eps}$ and $\sqrt{d}$ are integers, $w$ is also an integer.} to $w := \frac{(r-1)k\sqrt{d}}{2\eps}$ and for each $i$, the number of terms in the DNF within $B_i$ will be set to $N_i := 2^w \cdot e^{-i} = 2^{\frac{(r-1)k\sqrt{d}}{2\eps}(1-o(1))}$. The DNF defined over $B_i$ will assign function values from $\{ i \bmod{(r-1)}, i \bmod{(r-1)} +1\}$. %Note that $\lfloor i/k \rfloor \in [r-1]$ for all $i \in [k(r-1)]$. 
The terms in each DNF will be chosen randomly from the following distribution. We think of terms as points $t \in \{0,1\}^d$ in the hypercube where another point $x$ \emph{satisfies} $t$ if $t \preceq x$, i.e. $t_i = 1$ implies $x_i = 1$. 

\begin{definition} [Term distribution] \label{def:term} A term $t \in \{0,1\}^d$ is sampled from the distribution $\cD_{\texttt{term}}$ as follows. Form a (multi)-set $S \subseteq [d]$ by choosing $w$ independent uniform samples from $[d]$. For each $a \in [d]$,  let $t_a := \mathbf{1}(a \in S)$. \end{definition}

\subsection{The Distributions $\Dyes$ and $\Dno$}

%We now sample a function from the distributions $\cD_{yes}$ and $\cD_{no}$ as follows. 

We now define the yes and no distributions over functions $f \colon \{0,1\}^d \to [r]$. For each $i \in [k(r-1)]$, choose terms $t^{i,1},\ldots,t^{i,N_i}$ i.i.d. from $\cD_{\texttt{term}}$ and let $\pmb{t} = \{t^{i,j} \colon i \in [k(r-1)]\text{, }j \in [N_i]\}$ denote the random set of all terms. Now, for each $i \in [k(r-1)]$ and $j \in [N_i]$, define the set
\begin{align} \label{eq:unique_buckets}
    U_{i,j} = \left\{x \in B_i \colon x \succeq t^{i,j} \text{ and } x \not\succeq t^{i,j'} \text{ for all } j' \neq j\right\}
\end{align}
of all points in the $i$'th block that satisfy the $j$'th term \emph{uniquely}. Let $U_i = \bigcup_{j \in [N_i]} U_{i,j}$ denote the set of points in $B_i$ that satisfy a unique term. The following claim is key to our result and motivates our choice of $w$ and $N_i$. We defer its proof to \Cref{sec:proof-of-highprob}.

\begin{claim} \label{clm:highprob_vertex} For any $i \in [k(r-1)]$, $j \in [N_i]$, and $x \in B_i$, we have 
\[
\frac{1}{45N_i} \leq \Pr_{\bt}[x \in U_{i,j}] \leq \frac{3}{N_i}\text{.}
\]
As a corollary, we have $\Pr_{\bt}[x \in U_{i}] \geq 1/45$. \end{claim}

Functions drawn from $\Dyes$ are generated as follows. For each $i \in [k(r-1)]$ choose a uniform random assignment 
\[
\pmb{\phi}_i \colon [N_i] \to \left\{i \bmod (r-1), i \bmod (r-1) + 1\right\} \text{ and let } \pmb{\phi} = (\pmb{\phi}_i \colon i \in [k(r-1)]) \text{.}
\]
For every $x \in B_i$ define
\[
    f_{\pmb{t},\pmb{\phi}}(x) = \begin{cases}
        i \bmod (r-1), & \text{if } \forall j \in [N_i] \text{, } x \not\succeq t^{i,j} \\
        i \bmod (r-1)+1, & \text{if } \exists j \neq j' \in [N_i]\text{, } x \succeq t^{i,j},t^{i,j'} \\
        \pmb{\phi}_i(j), & \text{if } x \in U_{i,j}\text{.}
        \end{cases}
\]

Functions drawn $\Dno$ are generated as follows. For each $i \in [k(r-1)]$ choose a uniform random function 
\[
\pmb{r}_i \colon U_i \to \left\{i \bmod (r-1),i \bmod (r-1)+1\right\} \text{ and let } \pmb{r} = (\pmb{r}_i \colon i \in [k(r-1)]) \text{.} 
\]
For each $x \in B_i$ define

\[
    f_{\pmb{t},\pmb{r}}(x) = \begin{cases}
       i \bmod (r-1), & \text{if } \forall j \in [N_i] \text{, } x \not\succeq t^{i,j} \\
        i \bmod (r-1)+1, & \text{if } \exists j \neq j' \in [N_i]\text{, } x \succeq t^{i,j},t^{i,j'}\\
       \pmb{r}_i(x), & \text{if } x \in U_i\text{.}
        \end{cases}
\]
For $x$ not belonging to any $B_i$: if $|x| < \frac{d}{2}$, then both the yes and no distributions assign value $0$ and if $|x| \geq \frac{d}{2} + \eps\sqrt{d}$, then both the yes and no distributions assign value $r-1$. \\

In summary, a function $f_{\pmb{t},\pmb{\phi}} \sim \Dyes$ assigns the same random value $\pmb{\phi}_i(j) \in \{i \bmod (r-1),i \bmod (r-1) +1\}$ to all points in $U_{i,j}$, which results in a $k$-monotone function, whereas a function $f_{\pmb{t},\pmb{r}} \sim \Dno$ assigns an i.i.d. uniform random $\{i \bmod (r-1), i \bmod (r-1) +1\}$-value to each point in $U_i$, resulting in a function that is far from being $k$-monotone. By construction, to detect any difference between these cases a tester will need to sample at least two points from the same $U_{i,j}$. \Cref{thm:LB_2sided_samples_rk} follows immediately from the following three lemmas.

\begin{lemma} \label{lem:yes} Every function in the support of $\cD_{yes}$ is $k$-monotone. \end{lemma}

\begin{proof} Consider any $f_{\bt,\pmb{\phi}}(x) \in \text{supp}(\Dyes)$. For each $a \in [k]$, consider the union over a contiguous strip of $r-1$ blocks formed by 
\[
Y_a := B_{a(r-1)} \cup B_{a(r-1) + 1} \cup \cdots \cup B_{(a+1)(r-1)-1} \text{.}
\]
Recall that if $|x| < d/2$, then $f_{\bt,\pmb{\phi}}(x) = 0$ and if $|x| \geq d/2 + \eps \sqrt{d}$, then $f_{\bt,\pmb{\phi}}(x) = r-1$. If $d/2 \leq |x| < d/2 + \eps\sqrt{d}$, then $x \in \bigcup_{a \in [k]} Y_a$. Therefore, it suffices to show that for any pair of comparable points $x \prec y \in Y_a$, we have $f_{\bt,\pmb{\phi}}(x) \leq f_{\bt,\pmb{\phi}}(y)$. Firstly, observe that by construction all points $z \in B_{a(r-1)+b}$ have function value $f_{\bt,\pmb{\phi}}(z) \in \{b,b+1\}$. Since $x \prec y$, if $x$ and $y$ are in different blocks, then $x \in B_{a(r-1)+b}$ and $y \in B_{a(r-1)+b'}$ where $b < b'$ and so the inequality is satisfied. Therefore, we may assume $x,y \in B_{a(r-1) + b}$ are in the same block. Since $x \prec y$, if $t \prec x$ for some term $t \in \text{supp}(\Dterm)$, then $t \prec y$ as well. I.e. the set of terms in $B_{a(r-1)+b}$ satisfied by $y$ is a superset of the set of terms in $B_{a(r-1)+b}$ satisfied by $x$. By construction, this implies $f_{\bt,\pmb{\phi}}(x) \leq f_{\bt,\pmb{\phi}}(y)$. \end{proof}

\begin{lemma} \label{lemma:no} For $f_{\pmb{t},\pmb{r}} \sim \cD_{no}$, we have
$\Pr_{\pmb{t},\pmb{r}}[\eps(f_{\pmb{t},\pmb{r}},\cM_{r,k}) = \Omega(\eps)] = \Omega(1)$. 
\end{lemma}

We prove \Cref{lemma:no} in \Cref{sec:far}.

\begin{lemma} \label{lem:distinguish} Given a collection of points $\pmb{x} = (x_1,\ldots,x_s) \in (\{0,1\}^d)^s$ and a function $f \colon \{0,1\}^d \to [r]$, let $(\pmb{x},f(\pmb{x})) = ((x_1,f(x_1)),\ldots,(x_s,f(x_s))))$ denote the corresponding collection of labeled examples. Let $\cE_{yes}$ and $\cE_{no}$ denote the distributions over $(\pmb{x},f(\pmb{x}))$ when $\pmb{x}$ consists of $s$ i.i.d. uniform samples and $f \sim \Dyes$ and $f \sim \Dno$, respectively. If $s \leq 2^{\frac{(r-1)k\sqrt{d}}{5\eps}}$, then the total variation distance between $\mathcal{E}_{\texttt{yes}}$ and $\mathcal{E}_{\texttt{no}}$ is $o(1)$. \end{lemma}

We prove \Cref{lem:distinguish} in \Cref{sec:distinguish}.

\subsection{Proof of \Cref{clm:highprob_vertex}} \label{sec:proof-of-highprob}

\begin{proof} Recall $w = \frac{(r-1)k\sqrt{d}}{2\eps}$, $N_i = 2^w \cdot e^{-i}$, the definition of $\cD_{\texttt{term}}$ from \Cref{def:term}, and the definition of $U_{i,j}$ from \cref{eq:unique_buckets}. Since $x \in B_i$ we have $|x| = \frac{d}{2} + \ell$ where $\frac{i\eps \sqrt{d}}{k(r-1)} \leq \ell < \frac{(i+1)\eps \sqrt{d}}{k(r-1)}$. Note that $\Pr_{t \sim \cD_{\texttt{term}}}[t \preceq x] = \left(|x|/d\right)^w$ since $t \preceq x$ iff the non-zero coordinates of $t$ are a subset of the non-zero coordinates of $x$. Therefore, we have
\begin{align*}
    \Pr_{\bt}[x \in U_{i,j}] &= \Pr_{t^{i,j}}[t^{i,j} \preceq x] \cdot \prod_{j' \in [N_i] \setminus \{j\}} \Pr_{t^{i,j'}}[t^{i,j'} \not\preceq x] = \left(|x|/d\right)^w \left(1 - \left(|x|/d\right)^w\right)^{N_i-1} \text{.}
\end{align*}
Note that the first term is upper bounded as
\begin{align*}
    (|x|/d)^w &\leq \left(\frac{\frac{d}{2} + \frac{(i+1) \cdot \eps \sqrt{d}}{k(r-1)}}{d}\right)^w = \frac{1}{2^w}\left(1 + \frac{2\eps}{k(r-1)\sqrt{d}} \cdot (i+1)\right)^w \leq \frac{e^{i+1+o(1)}}{2^w} \leq \frac{e^{1+o(1)}}{N_i}
\end{align*}
and this immediately implies the upper bound on $\Pr_{\bt}[x \in U_{i,j}]$. We can also lower bound this quantity by 
\begin{align*}
    (|x|/d)^w &\geq \left(\frac{\frac{d}{2} + \frac{i \cdot \eps \sqrt{d}}{k(r-1)}}{d}\right)^w = \frac{1}{2^w}\left(1 + \frac{2\eps}{k(r-1)\sqrt{d}} \cdot i \right)^w \geq \frac{e^{i-o(1)}}{2^w} \geq \frac{1}{e^{o(1)} N_i} \text{.}
\end{align*}
Now, combining our upper and lower bounds on $(|x|/d)^w$ yields
\begin{align*}
    \pr_{\bt}[x \in U_{i,j}] &\geq \frac{1}{e^{o(1)}N_i}\left(1-\frac{e^{1+o(1)}}{N_i}\right)^{N_i} \geq \frac{1}{e^{o(1)}N_i} e^{-(1+o(1)) \cdot e^{1+o(1)}} \geq \frac{1}{e^{e+1} N_i} \geq \frac{1}{45 N_i} \text{.}\end{align*}
\end{proof}

\subsection{$\Dyes$ and $\Dno$ are Hard to Distinguish: Proof of \Cref{lem:distinguish}} \label{sec:distinguish}

\begin{proof} Recall the definition of the set $U_{i,j}$ in \cref{eq:unique_buckets}. For $a \neq b \in [s]$, let $E_{ab}$ denote the event that $x_a$ and $x_b$ belong to the same $U_{i,j}$ for some $i \in [k(r-1)]$ and $j \in [N_i]$. Observe that conditioned on $\overline{\vee_{a,b}E_{ab}}$, the distributions $\mathcal{E}_{yes}$ and $\mathcal{E}_{no}$ are identical. Let $x,y \in \{0,1\}^d$ denote two i.i.d. uniform samples. We have
\begin{align} \label{eq:asd}
    \Pr[E_{ab}] = \Pr_{x,y,\bt}\left[\bigvee_{i,j} (x \in U_{i,j} \wedge y \in U_{i,j})\right] = \sum_{i,j} \Pr_{x,y,\bt} \left[x \in U_{i,j} \wedge y \in U_{i,j}\right] = \sum_{i,j} \Pr_{x,\bt}[x \in U_{i,j}]^2 
\end{align}
where the first step holds since the $U_{i,j}$'s are disjoint and the second step holds by independence of $x$ and $y$. Now, for a fixed $i \in [k(r-1)]$ and $j \in [N_i]$ we have the following: by \Cref{clm:highprob_vertex}, for $x \in B_i$ we have $\Pr_{\bt}[x \in U_{i,j}] \leq \frac{3}{N_i}$ and for $x \notin B_i$ we have $\Pr_{\bt}[x \in U_{i,j}] = 0$. Therefore $\Pr_{x,\bt}[x \in U_{i,j}] \leq \frac{3}{N_i}$. Therefore, the RHS of \cref{eq:asd} is bounded as
\begin{align*}
    \sum_{i,j} \Pr_{x,\bt}[x \in U_{i,j}]^2 = \sum_i N_i \cdot \Pr_{x,\bt}[x \in U_{i,j}]^2 \leq \sum_i \frac{9}{N_i} \leq rk \cdot \frac{9}{N_{k(r-1)-1}}%\leq 2^{-\frac{(r-1)\sqrt{d}}{2\eps}(1-o(1))}
\end{align*}
since the $N_i$'s are decreasing with respect to $i$. Therefore,

\[
d_{TV}(\mathcal{E}_{\texttt{yes}},\mathcal{E}_{\texttt{no}}) \leq \pr_{\pmb{x},\pmb{t}}\left[\bigvee_{a,b \in [s]} E_{ab}\right] \leq s^2 \cdot rk \cdot \frac{9}{N_{k(r-1)-1}} = o(1)
\]
since $N_{k(r-1)-1} = 2^{\frac{(r-1)k\sqrt{d}}{2\eps}(1-o(1))} = \omega(s^2 \cdot rk)$. \end{proof}

\subsection{Functions Drawn from $\Dno$ are Far from $k$-Monotone: Proof of \Cref{lemma:no}} \label{sec:far}

\begin{proof} We will use \Cref{clm:far}, restated below for the special case of $r$-valued functions over the hypercube. Recall that $\mrk$ is the set of $k$-monotone functions $f \colon \{0,1\}^d \to [r]$.

\begin{claim} \label{clm:far_hypercube} Let $f \colon \{0,1\}^d \to [r]$ and $k' \geq 3k$ be an integer. Let $\cC \subset (\{0,1\}^d)^{k'}$ be a collection of disjoint $k'$-alternating chains for $f$. Then 
\[
\eps(f,\cM_{r,k}) \geq \frac{1}{3 \cdot 2^d} \cdot \left| \bigcup_{C \in \cC} C \right| \text{.}
\]
\end{claim}

From the above claim, we can lower bound the distance to $k$-monotonicity of $f$ by showing that it contains a collection of disjoint $k'$-alternating chains where $k' \geq 3k$ whose union makes up an $\Omega(\eps)$-fraction of the hypercube. 

Recall $U_i = U_{i,1} \cup \cdots \cup U_{i,N_i} \subseteq B_i$ and note that $f_{\bt,\br} \sim \Dno$ takes values only from $\{i \bmod{(r-1)}, i \bmod{(r-1)}+1\}$ in $B_i$. In particular, for $a \in \{0,1,\ldots,r-2\}$, let 
\begin{align} \label{eq:Xa}
    X_a = B_{a} \cup B_{(r-1) + a} \cup B_{2(r-1) + a} \cup \cdots \cup B_{(k-1)(r-1) + a} = \bigcup_{i \in [k]} B_{i(r-1) + a}
\end{align}
denote the union of blocks $B_{\ell}$ where $\ell = a $ mod $(r-1)$. Note that all points $x \in X_a$ are assigned value $f_{\bt,\br}(x) \in \{a,a+1\}$. Moreover, this value is chosen uniformly at random when $x \in \bigcup_{i \in [k]} U_{i(r-1) + a}$, which occurs with probability $\geq 1/45$ by \Cref{clm:highprob_vertex}. Let $k'' := \frac{\eps\sqrt{d}}{r-1}$ and recall that we are assuming $rk \leq \frac{\eps\sqrt{d}}{24300}$ and so $k'' \geq 24300 k$. We first show there exists a large collection $\cC_a$ of length-$k''$ disjoint chains in $X_a$ for all $a \in \{0,1,\ldots,r-2\}$.

\begin{claim} \label{clm:chains} For every $a \in \{0,1,\ldots,r-2\}$, there exists a collection of vertex disjoint chains $\cC_a \subset (X_a)^{k''}$ in $X_a$ of length $k''$ of size $|\cC_a| \geq \Omega(\frac{2^d}{\sqrt{d}})$.  \end{claim}

\begin{mdframed}[backgroundcolor=blue!4,hidealllines=true]
\begin{proof} We start by showing that there is a large matching in the transitive closure of the hypercube from $L_{\frac{d}{2}}$ to $L_{\frac{d}{2} + \eps\sqrt{d}-1}$. Consider the bipartite graph $(U,V,E)$ where $U := L_{\frac{d}{2}}$, $V := L_{\frac{d}{2} + \eps\sqrt{d}-1}$, and $E := \{(x,y) \in U \times V \colon x \prec y\}$. Observe that vertices in $U$ have degree exactly $\Delta := {\frac{d}{2} \choose \eps\sqrt{d}-1}$ while vertices in $V$ have degree exactly ${\frac{d}{2} + \eps\sqrt{d}-1 \choose \eps\sqrt{d}-1} \geq \Delta$. Note also that $|V| = {d \choose \frac{d}{2} + \eps\sqrt{d}-1} \geq \Omega(\frac{2^d}{\sqrt{d}})$ by Stirling's approximation. We now use the following claim from \cite{BBH24}.

\begin{claim} [Claim 5.10 of \cite{BBH24}] \label{clm:matching} Let $(U,V,E)$ be a bipartite graph and $\Delta > 0$ be such that (a) each vertex $x \in U$ has degree exactly $\Delta$ and (b) each vertex $y \in V$ has degree at least $\Delta$. Then there exists a matching $M \subseteq E$ in $(U,V,E)$ of size $|M| \geq \frac{1}{2} |V|$. \end{claim}

By the above claim and the previous observations, there exist subsets $S \subseteq L_{\frac{d}{2}}$ and $T \subseteq L_{\frac{d}{2} + \eps\sqrt{d}-1}$ of size $|S| = |T| = \Omega(\frac{2^d}{\sqrt{d}})$ and a bijection $\phi \colon S \to T$ satisfying $x \prec \phi(x)$ for all $x \in S$. We now use the following routing theorem due to Lehman and Ron to obtain a collection of disjoint chains from $S$ to $T$.

\begin{theorem} [Lehman-Ron, \cite{LR01}] Let $a < b$ and $S \subseteq L_a$, $T \subseteq L_b$ where $m := |S| = |T|$. Moreover, suppose there is a bijection $\phi \colon S \to T$ satisfying $x \prec \phi(x)$ for all $x \in S$. Then there exist $m$ vertex disjoint paths from $S$ to $T$ in the hypercube. \end{theorem}

Now, invoking the above theorem on our bijection $\phi \colon S \to T$ yields a collection $P$ of $|P| \geq \Omega(\frac{2^d}{\sqrt{d}})$ vertex disjoint paths from $L_{\frac{d}{2}}$ to $L_{\frac{d}{2}+\eps\sqrt{d}-1}$. For each $a \in \{0,1,\ldots,r-2\}$, let $\cC_a$ denote the collection of chains formed by taking a path in $P$ and including only the vertices from $X_a$ (recall \cref{eq:Xa}). Note that the resulting chains in $\cC_a$ are of length $k'' = \frac{\eps\sqrt{d}}{r-1}$. This completes the proof of \Cref{clm:chains}. \end{proof}
\end{mdframed}

From \Cref{clm:chains}, we have $\cC_0, \cC_1, \ldots, \cC_{r-2}$ where each $\cC_a \subset (X_a)^{k''}$ is a collection of vertex disjoint chains of length $k'' \geq 24300 k$ of size $|\cC_a| \geq \Omega(\frac{2^d}{\sqrt{d}})$. Fix a chain $C = (x_1,x_2,\ldots,x_{k''}) \in \cC_a$. Let $A(C)$ be the random variable which denotes the max-length alternating sub-chain (recall \Cref{def:alt-chains}) of $C$ over a random $f_{\bt,\br} \sim \Dno$. Fix $x_j$ in the chain and suppose $x_j \in B_i \subseteq X_a$. By \Cref{clm:highprob_vertex}, $\Pr_{\bt}[x_j \in U_i] \geq 1/45$. Moreover, conditioned on $x_j \in U_i$, $f_{\bt,\br}(x_j)$ is chosen from $\{a,a+1\}$ uniformly at random. Thus, any step of the sequence
\[
(1,f_{\bt,\br}(x_2)-f_{\bt,\br}(x_1),f_{\bt,\br}(x_3)-f_{\bt,\br}(x_2),\ldots,f_{\bt,\br}(x_{k''})-f_{\bt,\br}(x_{k''-1}))
\]
is non-zero \emph{and} differs in sign from the previous non-zero step with probability at least $1/90$ and so $\Exp[A(C)] \geq k''/90$. I.e., $0 \leq \Exp[k''-A(C)] < k''(1- \frac{1}{90})$. Thus, using Markov's inequality we have
\begin{align} \label{eq:markov1}
    \Pr\left[A(C) < \frac{k''}{8100}\right] = \Pr\left[k''-A(C) > k''\left(1-\frac{1}{90}\right)\left(1+\frac{1}{90}\right)\right] \leq \frac{1}{(1+\frac{1}{90})} = 1-\frac{1}{91}\text{.}
\end{align}
Now, let $\cC = \cC_0 \cup \cC_1 \cup \cdots \cup \cC_{r-2}$ and let $Z: = |\{C \in \cC \colon A(C) \geq \frac{k''}{8100}\}|$. By \cref{eq:markov1} we have $\Exp[Z] \geq |\cC|/91$ and $0 \leq \Exp[|\cC|-Z] \leq |\cC|(1-\frac{1}{91})$. Again using Markov's inequality, we have
\begin{align} \label{eq:markov2}
    \Pr\left[Z < \frac{|\cC|}{8281}\right] = \Pr\left[|\cC|-Z > |\cC|\left(1-\frac{1}{91}\right)\left(1+\frac{1}{91}\right) \right] \leq \frac{1}{(1+\frac{1}{91})} = 1-\frac{1}{92} \text{.}
\end{align} 
Now, for $C \in \cC$ such that $A(C) \geq k''/8100$, let $C'$ be any $(k''/8100)$-alternating sub-chain of $C$. Let $\cC' = \{C' \colon C \in \cC \text{ such that } A(C) \geq k''/8100\}$ which is a collection of disjoint $(k''/8100)$-alternating chains for $f_{\bt,\br}$.
Now, recall that $k'' \geq 24300k$ and so $k''/8100 \geq 3k$. Thus, if $Z \geq |\cC|/8281$, then $|\cC'| \geq |\cC|/8281$ and so by \Cref{clm:far_hypercube} we have
\begin{align} \label{eq:dist_LB}
    \eps(f_{\bt,\br},\cM_{r,k}) \geq \frac{1}{3 \cdot 2^d} \left|\bigcup_{C' \in \cC'} C'\right| \geq \frac{1}{3 \cdot 2^d} \cdot |\cC'| \cdot \frac{k''}{8100} \geq \frac{k'' \cdot |\cC|}{201,228,300 \cdot 2^d} 
\end{align}
By \Cref{clm:chains} we have $|\cC| \geq (r-1)\cdot \Omega(\frac{2^d}{\sqrt{d}})$ and recall that $k'' = \frac{\eps \sqrt{d}}{r-1}$. Thus, the RHS of \cref{eq:dist_LB} is $\Omega(\eps)$. In conclusion, 
\[
\Pr_{\bt,\br}\left[\eps(f_{\bt,\br},\cM_{r,k}) \geq \Omega(\eps)\right] \geq \Pr\left[Z \geq \frac{|\cC|}{8281}\right] \geq \frac{1}{92}
\]
by \cref{eq:markov2} and this completes the proof of \Cref{lemma:no}. \end{proof}

\section{Learning Upper Bound over Product Spaces} \label{sec:main-monotone-k}

In this section we prove \Cref{thm:UB-learner}, our upper bound for learning measurable $k$-monotone functions in $\RR^d$. We restate the theorem below without any hidden logarithmic factors and for the case of $r=2$. The theorem for general $r \geq 2$ can then be obtained by replacing $\eps$ with $\eps/r$ and $\delta$ by $\delta/3r$ following the same approach we used to prove \Cref{thm:UB-hypercube} in \Cref{sec:UB-hypercube-proof}.

\begin{theorem} \label{thm:UB-formal} Given an arbitrary product measure $\mu$, there is a learning algorithm under $\mu$ which learns any measurable $k$-monotone function $f \colon \RR^d \to \{\pm 1\}$ to error $\eps$ with probability $1-\delta$ with time and sample complexity 
\begin{align} \label{eq:complexities}
\mathrm{poly}\log\left(\frac{1}{\delta}\right) \cdot \min \left\{ (d\log (dk/\eps))^{O\left(\frac{k}{\eps}\sqrt{d \log (dk/\eps)}\right)}, \left(\frac{dk}{\eps}\right)^{O(d)} \right\}
\end{align}
\end{theorem}

Our proof uses downsampling to reduce our learning problem over $\RR^d$ to learning over a hypergrid, $[N]^d$, under the uniform distribution with mild label noise. In \Cref{sec:downsampling} we synthesize the results from \cite{HarmsY22} which we borrow for our proof. In \Cref{sec:hypergrids} we give two learning results for hypergrids whose time complexities correspond to the two arguments inside the $\min$ expression in \cref{eq:complexities}. In \Cref{sec:UB-final-proof} we describe the learning algorithm and prove its correctness.

Throughout this section, let $\mu = \prod_{i=1}^d \mu_i$ be any product measure over $\mathbb{R}^d$ and let $N$ be a power of two satisfying $8 k d / \eps \leq N \leq 16 k d /\eps$.

\subsection{Reduction to Hypergrids via Downsampling} \label{sec:downsampling}

%Given two maps $\block \colon \mathbb{R}^d \to [N]^d$, $\blockpoint \colon [N]^d \to \mathbb{R}^d$ and a function $f \colon \mathbb{R}^d \to \{\pm 1\}$ we define the functions $f^{\block} \colon [N]^d \to \{\pm 1\}$ and $f^{\coarse} \colon \RR^d \to \{\pm 1\}$ as $f^{\block}(z) = f(\blockpoint(z))$ and $f^{\coarse}(x) = f(\blockpoint(\block(x)))$. We let $\block(\mu)$ denote the distribution over $[N]^d$ induced by sampling $x \sim \mu$ and then taking $\block(x)$.

The idea of downsampling is to construct a grid-partition of $\RR^d$ into $N^d$ blocks such that (a) the measure of each block under $\mu$ is roughly $N^{-d}$, and (b) the function $f$ we're trying to learn is constant on most of the blocks. Roughly speaking, this allows us to learn $f$ under $\mu$ by learning a proxy for $f$ over $[N]^d$ under the uniform distribution. The value of $N$ needed to achieve this depends on what \cite{HarmsY22} call the "block boundary size" of the function. Formally, the downsampling procedure constructs query access to maps $\block \colon \mathbb{R}^d \to [N]^d$ and $\blockpoint \colon [N]^d \to \mathbb{R}^d$ which have various good properties which we will spell out in the rest of this section. One should think of $\block$ as mapping each point $x \in \RR^d$ to the block of the grid-partition that $x$ belongs to and $\blockpoint$ as mapping each block to some specific point contained in the block. See \cite[Def 2.1]{HarmsY22} for a formal definition. Given these maps and a function $f \colon \mathbb{R}^d \to \{\pm 1\}$ we define the function $f^{\block} \colon [N]^d \to \{\pm 1\}$ as $f^{\block}(z) = f(\blockpoint(z))$. We let $\block(\mu)$ denote the distribution over $[N]^d$ induced by sampling $x \sim \mu$ and then taking $\block(x)$.

%$f^{\block}(\block(x)) = f(\blockpoint(\block(x)))$.

\begin{proposition} [Downsampling, \cite{HarmsY22}] \label{prop:HY} Let $f \colon \mathbb{R}^d \to \{0,1\}$ be a $k$-monotone function and $N,Q \in \mathbb{Z}^+$. Using 
\[
m:= O\left(\frac{NQ^2d^2}{\min(\delta,\eps)^2} \ln \left(\frac{Nd}{\delta}\right)\right)
\]
samples from $\mu = \mu_1 \times \cdots \times \mu_d$, there is a downsampling procedure that constructs query access to maps $\block \colon \mathbb{R}^d \to [N]^d$ and $\blockpoint \colon [N]^d \to \mathbb{R}^d$ such that with probability at least $1-\delta$ over the random samples, the following two conditions are satisfied:
\begin{enumerate}
    \item $\norm{\block(\mu) - \mathsf{unif}([N]^d)}_{\text{TV}} \leq \frac{\delta}{Q}$.
    \item $\Pr_{x \sim \mu}\left[f(x) \neq f^{\block}(\block(x))\right] \leq \eps$.
\end{enumerate}
The total running time and number of samples is $O(m)$.
\end{proposition}
\begin{proof} \cite[Prop. 2.5]{HarmsY22} shows that there is a randomized procedure using $m$ samples from $\mu$ and $O(m)$ time which constructs the maps $\block$ and $\blockpoint$ such that with probability $1$, we get
\begin{align}
    \Pr_{x \sim \mu}\left[f(x) \neq f^{\block}(\block(x))\right] \leq N^{-d} \cdot \mathsf{bbs}(f,N) + \norm{\block(\mu) - \mathsf{unif}([N]^d)}_{\text{TV}}
\end{align}
where $\mathsf{bbs}(f,N)$ is the $N$-block boundary size of $f$ \cite[Def. 2.4]{HarmsY22}, which is at most $k d N^{d-1}$ when $f$ is $k$-monotone \cite[Lemma 7.1]{HarmsY22}. Thus, the first of the two quantities in the RHS is at most $kd/N$ which is at most $\eps/8$ using our definition of $N$. Then, \cite[Lemma 2.7]{HarmsY22} states that 
\begin{align}
    \Pr\left[\norm{\block(\mu) - \mathsf{unif}([N]^d)}_{\text{TV}} > \beta \right] \leq 4 N d \cdot \exp\left(-\frac{\beta^2 m}{18Nd^2}\right)
\end{align}
and so invoking this lemma with $\beta := \min(\delta/4Q,\eps/8)$ and $m := \frac{18Nd^2}{\beta^2} \ln \left(\frac{16Nd}{\delta}\right)$ completes the proof. \end{proof}

\subsection{Learning over Hypergrids} \label{sec:hypergrids}

\Cref{prop:HY} gives us a way to learn the function $f$ over $\RR^d$ with respect to $\mu$ by learning the function $f^{\block}$ over $[N]^d$ with respect to $\block(\mu)$ (which is close to uniform) under mild \emph{agnostic label noise}, since we have imperfect query access to $f^{\block}$ (item (2)). Fortunately, this mild label noise and the fact that $\block(\mu)$ is not exactly uniform pose only minor technical hurdles.

We prove the following two upper bounds for learning over hypergrids under mild agnostic label noise. The bound in \Cref{lem:hypergrid-coupon} is relatively straightforward to prove using coupon collector arguments. We give a proof in \Cref{sec:coupon}. The bound in \Cref{lem:hypergrid-cube} relies on the agnostic version of the low-degree learning algorithm due to \cite{agnostichalfspaces08}.

\begin{lemma} [Agnostic Coupon Collecting Learner] \label{lem:hypergrid-coupon} Let $\delta \in (0,1)$, $\eta \in (0,1/8)$, and $N \in \mathbb{Z}^+$. Let $f \colon [N]^d \to \{\pm 1\}$, be any $k$-monotone function and let $D$ be a distribution over pairs $(x,y) \in [N]^d \times \{\pm 1\}$ with uniform marginal over $x$ and such that $\Pr_{(x,y) \sim D}[f(x) \neq y] \leq \eta$. There is an algorithm which, given any such $f$ and $D$, uses at most 
\[
O(N^d \log^2 (N^d/\eta\delta))
\]
samples from $D$ and returns $h \colon [N]^d \to \{\pm 1\}$, satisfying $\Pr_h[d(f,h) \leq 5\eta] \geq 1-\delta$. \end{lemma}

\begin{lemma} [Hypercube Mapping Learner] \label{lem:hypergrid-cube} Let $\eps, \delta \in (0,1)$, $\eta \in (0,1/2)$, and $N \in \mathbb{Z}^+$ be a power of two. Let $f \colon [N]^d \to \{\pm 1\}$ be a $k$-monotone function and let $D$ be a distribution over pairs $(x,y) \in [N]^d \times \{\pm 1\}$ with uniform marginal over $x$ and such that $\Pr_{(x,y) \sim D}[f(x) \neq y] \leq \eta$. There is an algorithm which, given any such $f$ and $D$, uses at most
\[
\mathrm{poly}(\log(1/\delta))\cdot(d\log N)^{O\left(\frac{k}{\eps}\sqrt{d \log N}\right)}
\]
%\[
%O\left(\frac{1}{\eps^2(1-2\eta)^2} + \log \frac{1}{\delta} \right)(d\log N)^{O\left(\frac{k}{\eps}\sqrt{d \log N}\right)}
%\]
samples from $D$ and returns $h \colon [N]^d \to \{\pm 1\}$, satisfying  $\Pr_h[d(f,h) \leq \eps + 8\eta] \geq 1-\delta$. \end{lemma}

\begin{proof} Let $b \colon [N] \to \{\pm 1\}^{\log N}$ denote the bijection which maps each element of $[N]$ to its bit representation. Let $\pmb{b} \colon [N]^d \to \{\pm 1\}^{d \log N}$ be defined as $\pmb{b}(x) = (b(x_1),\ldots,b(x_d))$. Given $f \colon [N]^d \to \{\pm 1\}$ define the function $f^{\cube} \colon \{\pm 1\}^{d \log N} \to \{\pm 1\}$ as $f^{\cube}(z) = f(\pmb{b}^{-1}(z))$.

\begin{fact} \label{fact:cube-alt} If $f$ is $k$-monotone over $[N]^d$, then $f^{\cube}$ is $k$-monotone over $\{\pm 1\}^{d \log N}$. \end{fact}

\begin{mdframed}[backgroundcolor=blue!4,hidealllines=true]
\begin{proof} Observe that if $\pmb{b}(x) \prec \pmb{b}(y)$ in $\{\pm 1\}^{d\log N}$, then $x \prec y$ in $[N]^d$. Thus, if $\pmb{b}(x_1) \prec \cdots \prec \pmb{b}(x_m)$ is an $m$-alternating chain for $f^{\cube}$, then $x_1 \prec \cdots \prec x_m$ is an $m$-alternating chain for $f$. Therefore, if $f^{\mathsf{cube}}$ is not $k$-monotone, then neither is $f$. \end{proof}
\end{mdframed}

Now, given \Cref{fact:cube-alt} and the bijection $\pmb{b} \colon [N]^d \to \{\pm 1\}^{d \log N}$, it suffices to provide a learning algorithm for $f^{\cube}$. This is achieved using the agnostic version of the low-degree algorithm due to \cite{agnostichalfspaces08}. We use the following theorem statement\footnote{Note that the statement of Observation 3 in \cite{agnostichalfspaces08} assumes $\delta$ is a constant. However, it is clear by inspecting the proof that it extends to arbitrary $\delta$ with an extra $\mathrm{poly}(\log(1/\delta))$ factor in the sample complexity. This is because the proof of Observation 3 works with the output of a single call to the standard low-degree learner, which has the desired sample complexity (see Fact 1 of \cite{agnostichalfspaces08}).} rephrased from \cite{agnostichalfspaces08}. 

\begin{theorem} [Observation 3 of \cite{agnostichalfspaces08}, restated] \label{thm:agnostic-lowdegree} Let $\cC$ be a concept class of Boolean functions over $\{\pm 1\}^d$ such that for some fixed positive integer $\tau$, all $f \in \cC$ satisfy $\sum_{S \subseteq [d] \colon |S| > \tau} \widehat{f}(S)^2 \leq \eps/2$. Then, there is an algorithm $\cA$ which, on any input $f \in \cC$ and any distribution $D$ over pairs $(x,y) \sim \{\pm 1\}^d \times \{\pm 1\}$ with uniform marginal over $x$, uses at most $d^{O(\tau)} \cdot \mathrm{poly}(\log 1/\delta))$ samples from $D$ and with probability $1-\delta$ outputs a hypothesis $h \colon \{\pm 1\}^d \to \{\pm 1\}$ such that $d(f,h) \leq \eps + 8\mathtt{opt}$ where $\mathtt{opt} = \min_{g \in \cC} \Pr_{(x,y) \sim D}[g(x) \neq y]$. \end{theorem}

Now, recall that we have access to distribution $D$ over pairs $(x,y) \in [N]^d \times \{\pm 1\}$ with uniform marginal over $x$ and such that $\Pr_{(x,y) \sim D}[f(x) \neq y] \leq \eta$. Since $f$ is $k$-monotone, we have 
\begin{align} \label{eq:opt-learn}
    \mathtt{opt} \leq \Pr_{(x,y) \sim D}[f(x) \neq y] \leq \eta \text{.}
\end{align}
We use the following Fourier concentration lemma due to \cite{BlaisCOST15} for $k$-monotone Boolean functions.

\begin{lemma} [\cite{BlaisCOST15}] \label{lem:Fourier-concentration} If $f \colon \{\pm 1\}^d \to \{\pm 1\}$ is $k$-monotone, then $\sum_{S \colon |S| > \frac{k\sqrt{d}}{\eps}} \widehat{f}(S)^2 \leq \eps$. \end{lemma}

Finally, using \cref{eq:opt-learn} and \Cref{lem:Fourier-concentration}, we can invoke \Cref{thm:agnostic-lowdegree} with $\tau = \frac{2k\sqrt{d \log N}}{\eps}$, concluding the proof of \Cref{lem:hypergrid-cube}. \end{proof}

\subsection{Putting it Together: Proof of \Cref{thm:UB-formal}} \label{sec:UB-final-proof}

\begin{proof} We now have all the tools to define the algorithm and prove its correctness. \\

\begin{algorithm}[H]
\label{alg:learner}
\SetAlgoLined
\textbf{Input:}
 $\eps,\delta \in (0,1)$ and access to examples from $EX(f,\mu)$ where $f \colon \RR^d \to \{\pm 1\}$ is $k$-monotone\;  
1. Let $N$ be a power of $2$ such that $\frac{8 k d}{\eps} \leq N \leq \frac{16 k d}{\eps}$. Let $\cA$ denote the learning algorithm for $k$-monotone functions $g \colon [N]^d \to \{\pm 1\}$ which has the smaller sample-complexity among the algorithms guaranteed by \Cref{lem:hypergrid-coupon} and \Cref{lem:hypergrid-cube}. Let $Q$ be the sample-complexity of $\cA$\;
2. Run the downsampling procedure of \Cref{prop:HY} to obtain the maps $\block$, $\blockpoint$, and access to the corresponding function $f^{\block} \colon [N]^d \to \{\pm 1\}$ \;
3. Obtain a set of $Q$ examples $S \in (\mathbb{R}^d \times \{\pm 1\})^Q$ from $(EX(f,\mu))^Q$\;
4. Let $S^{\block} = \{(\block(x),f(x)) \colon (x,f(x)) \in S\} \in ([N]^d \times \{\pm 1\})^Q$\;
5. Run $\cA$ using the sample $S^{\block}$, which returns a hypothesis $h^{\block} \colon [N]^d \to \{\pm 1\}$ for $f^{\block}$\;
\textbf{Return} the hypothesis $h \colon \RR^d \to \{\pm 1\}$ for $f \colon \RR^d \to \{\pm 1\}$ defined as $h(x) = h^{\block}(\block(x))$
 \caption{Learning algorithm for $k$-monotone functions under product measure $\mu$}
\end{algorithm}

\vspace{1cm}

Recall that given maps $\block \colon \mathbb{R}^d \to [N]^d$, $\blockpoint \colon [N]^d \to \mathbb{R}^d$, and a function $f \colon \mathbb{R}^d \to \{\pm 1\}$ we define the function $f^{\block} \colon [N]^d \to \{\pm 1\}$ as $f^{\block}(z) = f(\blockpoint(z))$. Recall that $\block(\mu)$ is the distribution over $\block(x) \in [N]^d$ when $x \sim \mu$. By \Cref{prop:HY}, step (2) of \Alg{learner} results in the following items being satisfied with probability at least $1-\delta$.
\begin{enumerate}
    \item $\norm{\block(\mu) - \mathsf{unif}([N]^d)}_{\text{TV}} \leq \frac{\delta}{Q}$.
    \item $\Pr_{x \sim \mu}\left[f(x) \neq f^{\block}(\block(x))\right] \leq \eps$.
\end{enumerate}

Firstly, by item (2), an example $(\block(x),f(x))$ where $x \sim \mu$, is equivalent to an example $(z,b) \sim D$ supported over pairs $[N]^d \times \{\pm 1\}$ such that (a) the marginal distribution of $z$ is $\block(\mu)$ and (b) $\Pr_{(z,b) \sim D}[f^{\block}(z) \neq b] \leq \eta = \eps$. In particular, the set $S^{\block} \in ([N]^d \times \{\pm 1\})^Q$ from step (4) of \Alg{learner} is distributed according to $D^Q$. Now, as stated, \Cref{lem:hypergrid-coupon} and \Cref{lem:hypergrid-cube} assume $\cA$ has access to $(x,y) \sim D_U$ where $x$ is \emph{uniformly distributed} over $[N]^d$. However, the following claim shows that since $\block(\mu)$ and $\mathsf{unif}([N]^d))$ are sufficiently close (item (1) above), the guarantees on $\cA$ from \Cref{lem:hypergrid-coupon} and \Cref{lem:hypergrid-cube} extend to the scenario when $\cA$ has sample access to $D^Q$ instead.

%Firstly, by item (2), an example $(\block(x),f(x))$ where $x \sim \mu$, is equivalent to an example $(z,b) \sim EX^{\eta}(f^{\block},\block(\mu))$ for some $\eta \leq \eps$. I.e. the set $S^{\block} \in ([N]^d \times \{\pm 1\})^Q$ from step (4) of \Alg{learner} is distributed according to $(EX^{\eta}(f^{\block},\block(\mu)))^Q$. Now, as stated, \Cref{lem:hypergrid-coupon} and \Cref{lem:hypergrid-cube} only hold when $\cA$ is given a sample from $(EX^{\eta}(f^{\block},\mathsf{unif}([N]^d)))^Q$. However, the following claim shows that since $\block(\mu)$ and $\mathsf{unif}([N]^d))$ are sufficiently close (item (1) above), the guarantees on $\cA$ from \Cref{lem:hypergrid-coupon} and \Cref{lem:hypergrid-cube} also hold when $\cA$ is given a sample from $(EX^{\eta}(f^{\block},\block(\mu)))^Q$.

\begin{claim} \label{clm:close} Let $D,D_U$ be two distributions supported over $(x,y) \in [N]^d \times \{\pm 1\}$ with the following properties: (i) the marginal distributions of $x$ for $D$ and $D_U$ are $\cD$ and $\mathsf{unif}([N]^d)$ such that $\norm{\cD - \mathsf{unif}([N]^d)}_{TV} \leq \gamma$, and (ii) for a fixed $x$, the conditional distribution on $y$ is the same for $D$ and $D_U$. Additionally, let $\cC \colon \cX \to \{\pm 1\}$ be a concept class and let $\cA$ be an algorithm which given any $f \in \cC$, $\eps,\delta \in (0,1)$, uses a sample from $D_U^Q$ and produces $h$ satisfying $\Pr_{x \sim \mathsf{unif}([N]^d)}[h(x) \neq f(x)] \leq \eps$ with probability at least $1-\delta$. Then, given a sample from $D^Q$, $\cA$ produces $h$ satisfying $\Pr_{x \sim \cD}[h(x) \neq f(x)] \leq \eps + \gamma$ with probability at least $1-(\delta+\gamma Q)$. \end{claim}

%\begin{claim} \label{clm:close} Let $\cC \colon \cX \to \{\pm 1\}$ be a concept class and let $\cA$ be an algorithm which given any $f \in \cC$, $\eps,\delta \in (0,1)$, and $\eta \in [0,1/2)$ uses a sample from $(EX^{\eta}(f,\mathsf{unif}([N]^d)))^Q$ and produces $h$ satisfying $\Pr_{x \sim \mathsf{unif}([N]^d)}[h(x) \neq f(x)] \leq \eps$ with probability at least $1-\delta$. If $\cD$ is a distribution over $[N]^d$ with $\norm{\cD - \mathsf{unif}([N]^d)}_{TV} \leq \gamma$, then given a sample from $(EX^{\eta}(f,\cD))^Q$, $\cA$ produces $h$ satisfying $\Pr_{x \sim \cD}[h(x) \neq f(x)] \leq \eps + \gamma$ with probability at least $1-(\delta+\gamma Q)$. \end{claim}

Note that, if $\cA$ is given a sample from $D_U^Q$, then by \Cref{lem:hypergrid-coupon} and \Cref{lem:hypergrid-cube}, we obtain a hypothesis with error at most $\eps + 8\eta \leq 9\eps$ with probability $1-\delta$. Therefore, using \Cref{clm:close} and item (1) above, given the sample from $D^Q$, if step (2) of \Alg{learner} succeeds, then with probability at least $1-2\delta$, step (5) produces $h^{\block}$ such that 
\[
\Pr_{z \sim \block(\mu)}[h^{\block}(z) \neq f^{\block}(z)] \leq 9\eps + \frac{\delta}{Q} \leq 10\eps \text{.}
\]
By the triangle inequality and using our definition of $h$ in the return statement of \Alg{learner}, we have
\begin{align}
    \Pr_{x \sim \mu}[h(x) \neq f(x)] &\leq \Pr_{x \sim \mu}[f(x) \neq f^{\block}(\block(x))] + \Pr_{x \sim \mu}[f^{\block}(\block(x)) \neq h^{\block}(\block(x))] \nonumber \\
    &= \Pr_{x \sim \mu}[f(x) \neq f^{\block}(\block(x))] + \Pr_{z \sim \block(\mu)}[f^{\block}(z) \neq h^{\block}(z)] \text{.}
\end{align}
The first term in the RHS is at most $\eps$ by item (2) above and the second term is at most $10\eps$ as we argued in the previous paragraph. Finally, adding up the failure probabilities of steps (2) and (5), we conclude that \Alg{learner} produces $h$ satisfying $\Pr_{x \sim \mu}[h(x) \neq f(x)] \leq 11\eps$ with probability at least $1-3\delta$. Therefore, running \Alg{learner} with rescaled inputs $\eps' = \eps/11$ and $\delta' = \delta/3$ yields the theorem as stated. \end{proof}

\subsubsection{Proof of \Cref{clm:close}}

\begin{proof} It is a well-known fact that for two distributions $\cD_1$ and $\cD_2$, the TV-distance between the corresponding product distributions satisfies $\norm{\cD_1^Q - \cD_2^Q}_{TV} \leq Q \norm{\cD_1 - \cD_2}_{TV}$ and thus we have 
\[
\norm{\cD^Q - \mathsf{unif}([N]^d)^Q}_{TV} \leq \gamma Q
\]
Given a set of $Q$ examples $S \in ([N]^d \times \{\pm 1\})^Q$, let $E(S)$ denote the event that the algorithm $\cA$ fails to produce a hypothesis with error at most $\eps$, after sampling $S$. First, note that conditioned on a sample $S$, the distributions over labels are the same, and therefore
\begin{align}
      \Pr_{S \sim D^Q}[E(S)] - \Pr_{S \sim D_U^Q}[E(S)] = \Pr_{S \sim \cD^Q}[E(S)] - \Pr_{S \sim \mathsf{unif}([N]^d)^Q}[E(S)] \text{.}
\end{align}
Using the definition of TV-distance we have
\begin{align}
    \Pr_{S \sim \cD^Q}[E(S)] - \Pr_{S \sim \mathsf{unif}([N]^d)^Q}[E(S)] \leq \norm{\cD^Q - \mathsf{unif}([N]^d)^Q}_{TV} \leq \gamma Q
\end{align}
and therefore
\begin{align}
     \Pr_{S \sim D^Q}[E(S)] \leq \Pr_{S \sim D_U^Q}[E(S)] + \gamma Q \leq \delta + \gamma Q
\end{align}
where we used $\Pr_{S \sim D_U^Q}[E(S)] \leq \delta$ by the assumption in the statement of the claim. Now, conditioned on $\neg E(S)$, we have that $\cA$ produces $h$ satisfying $\Pr_{x \sim \mathsf{unif}([N]^d)}[h(x) \neq f(x)] \leq \eps$. Again using our bound on the TV-distance, we have
\[
\Pr_{x \sim \cD}[h(x) \neq f(x)] - \Pr_{x \sim \mathsf{unif}([N]^d)}[h(x) \neq f(x)] \leq \norm{\cD - \mathsf{unif}([N]^d)}_{TV} \leq \gamma
\]
and so $\Pr_{x \sim \cD}[h(x) \neq f(x)] \leq \eps + \gamma$. \end{proof}

\section{Coupon Collecting Learner: Proof of \Cref{lem:hypergrid-coupon}} \label{sec:coupon}

We prove the following more general statement, from which \Cref{lem:hypergrid-coupon} is an immediate corollary.

\begin{lemma} [Agnostic Coupon Collector] \label{lem:coupon} Let $\delta \in (0,1)$, $\eta \in (0,1/8)$, and $B \in \mathbb{Z}^+$. Let $f \colon [B] \to \{\pm 1\}$, be any function and let $D$ be a distribution over pairs $(x,y) \in [B] \times \{\pm 1\}$ with uniform marginal over $x$ and such that $\Pr_{(x,y) \sim D}[f(x) \neq y] \leq \eta$. There is an algorithm which, given any such $f$ and $D$, uses at most $O(B \log^2 (B/\eta\delta))$ samples from $D$ and returns $h \colon [B] \to \{\pm 1\}$, satisfying $\Pr_h[d(f,h) \leq 5\eta] \geq 1-\delta$. \end{lemma}

\begin{proof} We are given samples from a distribution $D$ over pairs $(x,y) \in [B] \times \{\pm 1\}$ such that $\Pr_{(x,y) \sim D}[f(x) \neq y] \leq \eta$ and $x$ is distributed uniformly in $[B]$. Define $\eta(x) = \Pr_{(x,y)}[f(x) \ne y ~|~ x]$ to be the probability that $x$ is mislabeled when it is sampled. Then, we can write 
\[
\eta \geq \Pr_{(x,y) \sim D}[f(x) \neq y] = \Exp_{x \in [B]}[\eta(x)]
\]
and so by Markov's inequality the number of points $x \in [B]$ such that $\eta(x) > 1/4$ is at most $4\eta B$. Let $X_{\eta > 1/4} = \{x \in [B] \colon \eta(x) > 1/4\}$ and $X_{\eta\leq 1/4} = [B] \setminus X_{\eta > 1/4}$. As argued, we have $|X_{\eta > 1/4}| \leq 4\eta B$ and we will show that the labels of points in $X_{\eta\leq 1/4}$ can be learned using a simple majority vote. 

The learner is defined as follows. Take $s$ samples from $D$ and for each $x \in [B]$, let $m_{x}$ denote the number of times $x$ has been sampled. Let $m^+_x,m^{-}_x$ denote the number of times $x$ has been sampled with the label $+1,-1$ respectively. The learner outputs the hypothesis $h \colon [B] \to \{\pm 1\}$ defined by the majority vote, $h(x) = \mathsf{sgn}(m^+_x - m^-_x)$.

\begin{claim} \label{clm:maj} Let $x \in X_{\eta\leq 1/4}$ and suppose that $m_x \geq m := 8\ln(2/\beta)$. Then $\Pr[\mathsf{sgn}(m^+_x - m^-_x) \neq f(x)] \leq \beta$. \end{claim}

%\comment{Need to fix this proof. Labels are $\{\pm 1\}$, not $\{0,1\}$}

\begin{mdframed}[backgroundcolor=blue!4,hidealllines=true]
\begin{proof} Each label seen for $x$ is an independent $\{\pm 1\}$-valued random variable which is equal to $f(x)$ with probability $(1-\eta(x))$ where $\eta(x) \leq 1/4$. Thus, $\Exp[m^+_x - m^-_x] = m_x((1-\eta(x))f(x) + \eta(x)(-f(x))) = m_x(1-2\eta(x))f(x)$ and we have
\begin{align*}
    \Pr\left[\mathsf{sgn}(m^+_x - m^-_x) \neq f(x)\right] &\leq \Pr\left[\left|(m^+_x - m^-_x) - \Exp[m^+_x - m^-_x]\right| \geq m_x(1-2\eta(x))\right] \\
    &\leq 2\exp\left(-\frac{2m_x^2(1-2\eta(x))^2}{4m_x}\right) \leq 2\exp\left(-\frac{m_x^2}{8m_x}\right) \leq \beta
\end{align*}
by Hoeffding's inequality, the fact that $\eta(x) \leq 1/4$, and our bound on $m_x$. \end{proof}
\end{mdframed}

\begin{claim} \label{clm:batch} Suppose we take $s := B\ln(B/\alpha)$ samples. Then $\Pr[\exists x \in [B] \colon m_x = 0] < \alpha$. \end{claim}

\begin{mdframed}[backgroundcolor=blue!4,hidealllines=true]
\begin{proof} For any $x$, $\Pr[m_x = 0] = (1-\frac{1}{B})^s \leq \alpha/B$ and a union bound completes the proof. \end{proof}
\end{mdframed}

The following is an immediate corollary of \Cref{clm:batch}.

\begin{claim} \label{clm:m-coupons} Suppose we take $mB \cdot \ln(\frac{2mB}{\delta})$ samples. Then, $\Pr[\exists x \in [B] \colon m_x < m] < \delta/2$. \end{claim}

\begin{mdframed}[backgroundcolor=blue!4,hidealllines=true]
\begin{proof} Partition the samples into $m$ batches of size $B\ln(\frac{2mB}{\delta})$. Invoke \Cref{clm:batch} on each batch of samples with $\alpha := \frac{\delta}{2m}$. By the claim, each batch of samples contains a least $1$ copy of every point in $[B]$ with probability at least $1-\frac{\delta}{2m}$. Thus, by a union bound over the $m$ batches, our sample contains at least $m$ copies of every point in $[B]$ with probability at least $1-\frac{\delta}{2}$. \end{proof}
\end{mdframed}

%Let $m := \frac{2}{(1-2\eta)^2}\ln\left(\frac{4}{\eps\delta}\right)$ and $s := m \ln(\frac{2m}{\delta}) N^{2d}$. The learner takes $s$ samples from $EX^{\eta}(f,\mathsf{unif}([N]^d))$. Let $\cE$ denote the event that $m_x \geq m$ for all $x \in [N]^d$. By \Cref{clm:m-coupons}, we have $\Pr[\cE] \geq 1-\frac{\delta}{2}$. For each $x \in [N]^d$, let $\cB_x = \mathbf{1}(\mathsf{sgn}(m_x^+-m_x^-) \neq f(x))$, i.e. the indicator that $x$ is misclassified by the learner. 

Let $m := 8\ln\left(4/\eta\delta\right)$ and $s := mB \ln(2mB/\delta)$. The learner takes $s$ samples from $D$. Let $\cE$ denote the event that $m_x \geq m$ for all $x \in [B]$. By \Cref{clm:m-coupons}, we have $\Pr[\cE] \geq 1-\frac{\delta}{2}$. For each $x \in [B]$, let $\cB_x = \mathbf{1}(\mathsf{sgn}(m_x^+-m_x^-) \neq f(x))$ be the indicator that $x$ is misclassified by the learner. By \Cref{clm:maj} (using $\beta := \eta\delta/2$), we have
\begin{align}
    \Exp\left[ \sum_{x \in X_{\eta \leq 1/4}} \cB_x ~\Big|~ \cE \right] \leq \frac{\eta\delta}{2} B \text{ and thus } \Pr\left[ \sum_{x \in X_{\eta \leq 1/4}} \cB_x > \eta B ~\Big|~ \cE\right] \leq \frac{\delta}{2}
\end{align}
by Markov's inequality. Therefore,
\begin{align}
    \Pr\left[\sum_{x \in X_{\eta \leq 1/4}} \cB_x > \eta B\right] &= \Pr\left[ \sum_{x \in X_{\eta \leq 1/4}} \cB_x > \eta B ~\Big|~ \cE\right]\Pr[\cE] + \Pr\left[ \sum_{x \in X_{\eta \leq 1/4}} \cB_x > \eta B ~\Big|~ \neg \cE\right]\Pr[\neg \cE] \\
    &\leq \Pr\left[ \sum_{x \in X_{\eta \leq 1/4}} \cB_x > \eta B ~\Big|~ \cE\right] + \Pr[\neg \cE]
\end{align}
which is at most $\delta$. Finally, since $|X_{\eta > 1/4}| \leq 4\eta B$ (as argued above), even if all of $X_{\eta > 1/4}$ is misclassified, the total number of misclassified points is at most $5 \eta B$ with probability at least $1-\delta$. The number of examples used by the learner is
\[
s = m B \ln\left(\frac{2m B}{\delta}\right) = O(B \ln^2(B/\eta\delta))
\]
and this completes the proof. \end{proof}

\section{Testing by Learning} \label{sec:testing-by-learning}

%We require a testing-by-learning reduction that is slightly nonstandard, because the learning algorithm is not proper, i.e. for a hypothesis class $\cF$ it does not necessarily output a function in $\cF$.
%Given a domain $\cX$ and two functions $f,g \colon \cX \to \{\pm 1\}$, let $\dist(f,g) = \pr_{x \sim \cX}[h(x) \neq f(x)]$. The following lemma is specialized to the uniform distribution, which is sufficient for our purposes. The observation is not new; it appears in \cite{CFSS17} and possibly in other places.

\begin{lemma} \label{lem:learning->testing} Let $\cX$ be a domain, let $\mu$ be a measure over $\cX$, and let $\cF \colon \cX \to \{\pm 1\}$ be a class of Boolean-valued functions over $\cX$. Suppose that for every $\eps \in (0,1)$ there exists a learning algorithm $L$ for $\cF$ under $\mu$ using $s(\eps)$ samples. Then for every $\eps \in (0,1)$ there is an $\eps$-tester for $\cF$ under $\mu$ using $s(\eps/4) + O(1/\eps^2)$ samples.
\end{lemma}

\begin{proof} In the following, for functions $f,g \colon \cX \to \{\pm 1\}$, let $d(f,g) = \pr_{x \sim \mu}[f(x) \neq g(x)]$ and $d(f,\cF) = \min_{g \in \cF} d(f,g)$. We define the property testing algorithm $T$ as follows.
\begin{enumerate}
    \item Take $s := s(\eps/4,\cX)$ samples $(x_{1},f(x_1)),\ldots,(x_{s},f(x_s))$ and run $L$ to obtain a hypothesis $h$ for $f$.
    \item Compute a function $g \in \cF$ for which $d(h,g) = d(h,\cF)$. (We remark that this step incurs a blowup in time-complexity, but does not require any additional samples.)
    \item Take $s' := \frac{20}{\eps^2}$ new samples $(z_{1},f(z_1)),\ldots,(z_{s'},f(z_{s'}))$ and let $\alpha := \frac{1}{s'} \sum_{i=1}^{s'} \mathbf{1}(g(z_i) \neq f(z_i))$ be an empirical estimate for $d(f,g)$.
    \item If $\alpha \leq \frac{3\eps}{4}$, then accept. If $\alpha > \frac{3\eps}{4}$, then reject.
\end{enumerate}

\begin{claim} \label{clm:close-vs-far} If $f \in \cF$, then $\pr[d(f,g) \leq \eps/2] \geq 5/6$. \end{claim}

\begin{mdframed}[backgroundcolor=blue!4,hidealllines=true]
\begin{proof} By the guarantee of the learning algorithm, we have $\pr[d(f,h) \leq \eps/4] \geq 5/6$. Now, since $g$ is a function in $\cF$ as close as possible to $h$, we have $d(h,g) \leq d(h,f)$. Thus, if $d(f,h) \leq \eps/4$, then $d(h,g) \leq \eps/4$ as well. Thus, by the triangle inequality, with probability at least $5/6$ we have $d(f,g) \leq d(f,h) + d(h,g) \leq \eps/2$ as claimed. \end{proof}
\end{mdframed}

Now, consider the quantity $\alpha$ from step (4) of the algorithm, $T$. Let $X$ be the Bernoulli random variable which equals $1$ with probability $d(f,g)$. Note that $\alpha = \frac{1}{s'}\sum_{i=1}^{s'} X_i$ where the $X_i$'s are independent copies of $X$. Using Hoeffding's inequality we have
\begin{align*}
    \pr\left[|\alpha - d(f,g)| \geq \frac{\eps}{4}\right] = \pr\left[\left|\sum_{i=1}^{s'} X_i - s' \cdot d(f,g)\right| \geq \frac{s' \eps}{4}\right] \leq 2\exp\left(-\frac{2 \cdot (s'\eps/4)^2}{s'}\right) = \frac{2}{e^{\frac{s' \eps^2}{8}}}
\end{align*}
which is at most $1/6$ when $s' \geq 8\ln(12)/\eps^2$. We can now argue that the tester $T$ succeeds with probability at least $2/3$. There are two cases to consider. 

\begin{enumerate}
    \item $f \in \cF$: By \Cref{clm:close-vs-far}, $d(f,g) > \eps/2$ with probability less than $1/6$ and by the above calculation $|\alpha - d(f,g)| \geq \eps/4$ with probability at most $1/6$. By a union bound, with probability at least $2/3$ neither event occurs, and conditioned on this we have $\alpha \leq d(f,g) + \eps/4 \leq 3\eps/4$ and the algorithm accepts.
    \item $d(f,\cF) \geq \eps$: Then $d(f,g) \geq \eps$ since $g \in \cF$. Again, $|\alpha - d(f,g)| < \eps/4$ with probability at least $5/6$ and conditioned on this event occurring we have $\alpha > d(f,g) - \eps/4 \geq 3\eps/4$ and the algorithm rejects.
\end{enumerate}

Therefore, $T$ satisfies the conditions needed for \Cref{lem:learning->testing}. \end{proof}

\section{Sample-Based Testing with One-Sided Error}

In this section we prove \Cref{thm:one-sided}, our upper and lower bound on sample-based testing with one-sided error over the hypercube.

\begin{proof}[Proof of \Cref{thm:one-sided}] By a coupon-collecting argument, there is an $O(d \cdot 2^d)$ sample upper bound for \emph{exactly learning} any function over $\{0,1\}^d$ under the uniform distribution and therefore the upper bound is trivial.

It suffices to prove the lower bound for the case of $r = 2$ and $k=1$, i.e. for testing monotonicity of Boolean functions. We will need the following fact.

\begin{fact} \label{fact:shatter} Let $A \subset \{0,1\}^d$ be any anti-chain\footnote{A set $A \subset \{0,1\}^{d}$ is an \emph{anti-chain} if $x \not\preceq y$ for every $x \neq y \in A$. That is, no two points in $A$ are comparable under the partial order.} and let $\ell \colon A \to \{0,1\}$ be any labeling of $A$. Then there exists a monotone function $f \colon \{0,1\}^d \to \{0,1\}$ such that $f(x) = \ell(x)$ for all $x \in A$. That is, $A$ \emph{shatters} the class of monotone functions. \end{fact}

Now, let $T$ be any monotonicity tester with one-sided error and let $S \subseteq \{0,1\}^d$ denote a set of $s$ i.i.d. uniform samples. Since $T$ has one-sided error, if the input function is monotone, then $T$ must accept. In other words, for $T$ to reject it must be sure without a doubt that the input function is not monotone. By \Cref{fact:shatter} for $T$ to be sure the input function is not monotone, it must be that $S$ is \emph{not} an anti-chain. Let $f \colon \{0,1\}^d \to \{0,1\}$ be any function which is $\eps$-far from monotone. Since $T$ is a valid tester, it rejects $f$ with probability at least $2/3$. By the above argument we have
\begin{align} \label{eq:one-sided1}
    2/3 \leq \pr_{S}[T \text{ rejects } f] \leq \pr_{S}[S \text{ is not an anti-chain}] \leq s^2 \cdot \pr_{x,y \sim \{0,1\}^d}[x \preceq y]
\end{align}
where the last inequality is by a union bound over all pairs of samples. We then have   
\begin{align} \label{eq:one-sided2}
    \pr_{x,y \sim \{0,1\}^d}[x \preceq y] = \pr_{x,y \sim \{0,1\}^d}[x_i \leq y_i \text{, } \forall i \in [d]] = \prod_{i=1}^d \pr_{x_i,y_i \sim \{0,1\}}[x_i \leq y_i] = (3/4)^d \text{.}
\end{align}
Thus, combining \cref{eq:one-sided1} and \cref{eq:one-sided2} yields $s \geq \sqrt{\frac{2}{3} (\frac{4}{3})^d} = \exp(\Omega(d))$. \end{proof}

\section{Acknowledgements} 

We would like to thank Eric Blais and Nathaniel Harms for helpful discussions during the early stages of this work and for their thoughtful feedback. We would also like to thank the anonymous reviewers whose comments helped significantly to improve this write up. In particular, we are especially grateful to an anonymous journal reviewer who pointed out a critical inconsistency in the noise model being used in \Cref{sec:main-monotone-k} in a previous version of the paper. 

\bibliographystyle{alpha}
\bibliography{biblio_main}

\end{document}

%% file: preamble.tex
	\usepackage{a4,geometry}
\usepackage{graphicx}
\usepackage{amsmath,amssymb,amsthm,mathtools}
\usepackage{paralist}
\usepackage{bm}
\usepackage{bbm}
\usepackage{xspace}
\usepackage{url}
\usepackage{fullpage, prettyref}
\usepackage{boxedminipage}
\usepackage{wrapfig}
\usepackage{ifthen}
\usepackage{color}
\usepackage{xcolor}
\usepackage{framed}
\usepackage[linktocpage=true,pagebackref,letterpaper=true,colorlinks=true,pdfpagemode=none,urlcolor=blue,linkcolor=blue,citecolor=violet,pdfstartview=FitH]{hyperref}
\usepackage[nameinlink]{cleveref}
\usepackage{fullpage}
\usepackage{esvect}
\usepackage{mdframed}
\usepackage{thmtools}
\usepackage{thm-restate}
\newtheorem{theorem}{Theorem}[section]

\newtheorem{lemma}[theorem]{Lemma}
\newtheorem{claim}[theorem]{Claim}
\newtheorem{corollary}[theorem]{Corollary}
\newtheorem{definition}[theorem]{Definition}
\newtheorem{proposition}[theorem]{Proposition}
\newtheorem{fact}[theorem]{Fact}

\newtheorem{question}[theorem]{Question}

\newcommand{\ignore}[1]{}

\newcommand{\cA}{{\cal A}}
\newcommand{\cB}{\mathcal{B}}
\newcommand{\cC}{{\cal C}}
\newcommand{\cD}{\mathcal{D}}
\newcommand{\cE}{{\cal E}}
\newcommand{\cF}{\mathcal{F}}

\newcommand{\cM}{{\cal M}}
\newcommand{\cP}{\mathcal{P}}

\newcommand{\cX}{{\cal X}}

\newcommand{\eps}{\varepsilon}

\newcommand{\br}{\mathbf{r}}

\newcommand{\bt}{\mathbf{t}}

\newcommand{\RR}{\mathbb{R}}

\newcommand{\norm}[1]{\left\lVert #1 \right\rVert}

\newcommand{\Exp}{\EX}

\newcommand{\Sec}[1]{\hyperref[sec:#1]{\S\ref*{sec:#1}}} %section
\newcommand{\Eqn}[1]{\hyperref[eq:#1]{(\ref*{eq:#1})}} %equation
\newcommand{\Fig}[1]{\hyperref[fig:#1]{Fig.\,\ref*{fig:#1}}} %figure
\newcommand{\Tab}[1]{\hyperref[tab:#1]{Tab.\,\ref*{tab:#1}}} %table
\newcommand{\Thm}[1]{\hyperref[thm:#1]{Theorem\,\ref*{thm:#1}}} %theorem
\newcommand{\Fact}[1]{\hyperref[fact:#1]{Fact\,\ref*{fact:#1}}} %fact
\newcommand{\Lem}[1]{\hyperref[lem:#1]{Lemma\,\ref*{lem:#1}}} %lemma
\newcommand{\Prop}[1]{\hyperref[prop:#1]{Prop.~\ref*{prop:#1}}} %property
\newcommand{\Cor}[1]{\hyperref[cor:#1]{Corollary~\ref*{cor:#1}}} %corollary
\newcommand{\Conj}[1]{\hyperref[conj:#1]{Conjecture~\ref*{conj:#1}}} %conjecture
\newcommand{\Def}[1]{\hyperref[def:#1]{Definition~\ref*{def:#1}}} %definition
\newcommand{\Alg}[1]{\hyperref[alg:#1]{Alg.~\ref*{alg:#1}}} %algorithm
\newcommand{\Ex}[1]{\hyperref[ex:#1]{Ex.~\ref*{ex:#1}}} %example
\newcommand{\Clm}[1]{\hyperref[clm:#1]{Claim~\ref*{clm:#1}}} %example
\newcommand{\Step}[1]{\hyperref[step:#1]{Step~\ref*{step:#1}}} %example